\documentclass[11pt]{article}
\usepackage{amsmath,amssymb,amsthm} 
\usepackage{bbm}
\usepackage{stmaryrd}
\usepackage{paralist}
\usepackage{hyperref}
\usepackage{color}

\usepackage{multirow}
\setdefaultenum{(i)}{}{}{}
\usepackage[margin=1in]{geometry}
\usepackage[pdftex]{graphicx}
\usepackage{subfigure}
\usepackage{cases}
\theoremstyle{plain}
\newtheorem{mythm}{Theorem} \numberwithin{mythm}{section}
\newtheorem{myprop}[mythm]{Proposition}
\newtheorem{mylemma}[mythm]{Lemma}
\newtheorem{mydef}[mythm]{Definition}

\newtheorem{myassump}[mythm]{Assumption}
\newtheorem{mycor}[mythm]{Corollary}

\DeclareMathAlphabet\scr{U}{scr}{m}{n}
\SetMathAlphabet\scr{bold}{U}{scr}{b}{n}
  \DeclareFontFamily{U}{scr}{\skewchar\font'177}%
  \DeclareFontShape{U}{scr}{m}{n}{<-6>rsfs5<6-8>rsfs7<8->rsfs10}{}%
  \DeclareFontShape{U}{scr}{b}{n}{<-6>rsfs5<6-8>rsfs7<8->rsfs10}{}%

\numberwithin{equation}{section}

\renewenvironment{thebibliography}[1]{%
\begin{oldthebibliography}{#1}%
\setlength{\baselineskip}{.9em}
\linespread{.9}
\small
\setlength{\parskip}{0ex}%
\setlength{\itemsep}{.1em}%
}%
{%
\end{oldthebibliography}%
}

\DeclareMathOperator{\R}{\mathbb{R}}

\DeclareMathOperator{\TC}{TC}
\DeclareMathOperator{\TAC}{TAC}
\DeclareMathOperator{\DE}{DE}
\DeclareMathOperator{\E}{\mathbb{E}}

\def\vep{\varepsilon}
\def\dbR{\mathbb{R}}

\def\dbP{\mathbb{P}}

\begin{document}
\title{Optimal Rebalancing Frequencies for Multidimensional Portfolios\footnote{For helpful comments, we thank Kai Du, Paolo Guasoni, Yaroslav Melnyk, and Max Reppen. We are also indebted to latter for generously sharing his research results. Moreover, the pertinent remarks of an anonymous referee are gratefully acknowledged.} }
\author{Ibrahim Ekren\thanks{ETH Z\"urich, Departement Mathematik, R\"amistrasse 101, CH-8092, Z\"urich, Switzerland, email:
\texttt{ibrahim.ekren@math.ethz.ch}.}
\and
Ren Liu\thanks{ETH Z\"urich, Departement Mathematik, R\"amistrasse 101, CH-8092, Z\"urich, Switzerland, email:
\texttt{ren.liu@math.ethz.ch}.} 
\and
Johannes Muhle-Karbe\thanks{University of Michigan, Department of Mathematics, 530 Church Street, Ann Arbor, MI 48109, USA, email \texttt{johanmk@umich.edu}. Parts of this research were completed while this author was visiting the Forschungsinstitut f\"ur Mathematik (FIM) at ETH Z\"urich.}
}

\date{May 24, 2017}
\pagestyle{plain}
\maketitle
%-------------------------------------------------------------------------------------------------------------------------------------------------------------------------------
%------------------------------------------------ ABSTRACT---------------------------------------------------------------------------------------------------------------------
%-------------------------------------------------------------------------------------------------------------------------------------------------------------------------------
\begin{abstract}
We study optimal investment with multiple assets in the presence of small proportional transaction costs. Rather than computing an asymptotically optimal no-trade region, we optimize over suitable trading frequencies. We derive explicit formulas for these and the associated welfare losses due to small transaction costs in a general, multidimensional diffusion setting, and compare their performance to a number of alternatives using Monte Carlo simulations.
\end{abstract}

\noindent\textbf{Mathematics Subject Classification: (2010)}: 91G10, 91G80.

\noindent\textbf{JEL Classification:} G11. 

\noindent\textbf{Keywords:} transaction costs, optimal trading frequency, optimal investment, multiple assets

%-------------------------------------------------------------------------------------------------------------------------------------------------------------------------------
%----------------------------------------------------INTRO-------------------------------------------------------------------------------------------------------------------
%-------------------------------------------------------------------------------------------------------------------------------------------------------------------------------
\section{Introduction}
A basic problem in asset management is when and how to rebalance portfolios. In doing so, traders need to strike a balance between closely tracking a frictionless target portfolio that implements the optimal risk-return tradeoff, and limiting transaction costs, which make all too frequent trades prohibitively expensive. 

In practice, this issue is often addressed by specifying a ``suitable'' trading frequency, daily, weekly, or monthly rebalancing say, in a more or less ad hoc manner. Theory, however, points out that this approach is not optimal. Indeed, for proportional transaction costs -- a reasonable assumption for mid-sized investors -- there is large body of literature documenting that optimal strategies should be ``move-'' rather than ``time-based'' \cite{magill.constantinides.76,constantinides.86,davis.norman.90,shreve.soner.94,whalley.wilmott.97}. To wit, rebalancing times should not be chosen exogenously, but determined endogenously by the excursions of the investors' actual portfolios from their frictionless target. In recent years, there has been substantial progress on problems of this type, leading to a rather complete analysis and explicit rebalancing rules in the limiting regime of small transaction costs for the case of a \emph{single} risky asset \cite{soner.touzi.13,rosenbaum.tankov.14,martin.12,kallsen.muhlekarbe.15}. 

In contrast, much less is known about the practically very important case of \emph{many} risky assets. There are some numerical results for the bivariate Black-Scholes model~\cite{muthuraman.kumar.06}, but despite some recent progress in asymptotic analysis \cite{bichuch.shreve.13,possamai.al.15}, problems of this kind have remained rather intractable.\footnote{Models with fixed or quadratic trading costs are more tractable, see \cite{altarovici.al.15} or \cite{garleanu.pedersen.13a,dufresne.al.15,moreau.al.15,guasoni.weber.15}, respectively, but the issue is only exacerbated for nonlinear transaction costs such as the ones derived from the square-root price impact advocated by many practitioners \cite{almgren.al.05,toth.al.11}.}

In the present study, we therefore revisit the simpler time-based approach. In a general multidimensional diffusion setting, we consider an investor who maximizes her expected relative returns, penalized for the corresponding (local) variances and transaction costs.\footnote{Related ``local'' criteria are used in \cite{rosenbaum.tankov.14,martin.12,garleanu.pedersen.13a,guasoni.mayerhofer.14}. Under suitable integrability conditions (compare, e.g., \cite{kallsen.li.13}), our asymptotic results could be extended to global criteria.} In the limit for small transaction costs, we explicitly determine the optimal ``time-based rebalancing rule'', for which the next trading time is already specified when the current trade is implemented.\footnote{A special case are the constant trading frequencies considered in the previous literature, e.g., \cite{bertsimas.al.00,hayashi.mykland.05}; more generally, the waiting times can depend on current market characteristics.} We also provide an explicit asymptotic formula quantifying the performance loss compared to the frictionless case. These results allow to shed new light on the rebalancing of (multidimensional) portfolios in a number of ways.

First, for a single risky asset, they allow us to explicitly quantify the suboptimality of time-based relative to move-based rebalancing rules \cite{janecek.shreve.04,martin.12,soner.touzi.13,kallsen.muhlekarbe.15}. It turns out the the ratio of the corresponding utility losses compared to the frictionless case is a universal constant, independent of market or preference parameters: $(12/\pi)^{1/3} \approx 1.56.$ This number should be compared to the constant $2^{1/3} \approx 1.26$ differentiating the optimal move-based strategy~\cite{guasoni.mayerhofer.14} (which is implemented by reflection off its trading boundaries) from the optimal strategy within the class of strategies that trade back immediately to the frictionless target once these boundaries are breached \cite{melnyk.seifried.14}. Whence, passing to the move-based strategy introduces an additional ``suboptimality factor'' of almost the same size ($\approx 1.24$) as trading back immediately to the frictionless target ($\approx 1.26$).

Second, we provide an explicit formula for the performance of a simple benchmark strategy in a general multidimensional setting. In the one-dimensional case, the crucial statistic for the welfare losses due to transaction costs is the diffusion coefficient of the frictionless target weight \cite{whalley.wilmott.97,janecek.shreve.04,kallsen.muhlekarbe.15}. In our setting, this role is played by the diffusion matrix of the difference between the frictionless target and its buy-and-hold approximation. This quantity determines both the optimal waiting times and the performance of the corresponding policy through its  $L_{2,1}$-norm defined in \eqref{eq:L21} and, weighted with the risky assets' diffusion matrix, through its trace norm. 

Third, we perform a number of numerical tests that compare the performance of our time-based trading rule with various alternatives. For Black-Scholes models with constant investment opportunities, we find that the performance of the optimal time-based rebalancing rule virtually coincides with the optimal move-based portfolio, both for one and for two risky assets. In these settings, the optimal trading frequency for a 1\% transaction cost is less than every two years, so that the particular method for rebalancing does not play a crucial role. To illustrate how this changes in models with stochastic investment opportunities, we consider a truncated version of the model of Kim and Omberg~\cite{kim.omberg.96}, where expected returns are mean-reverting and optimal strategies are of trend-following type. Here, the optimal trading frequency increases to around once every half year, and optimal move-based strategies offer a marked improvement when they can be computed explicitly for a single risky asset. With more than one risky asset, the no-trade regions characterizing optimal move-based strategies are no longer known explicitly, and numerical algorithms for their computation are not available either. However, we find that even for correlated risky assets, a concatenation of univariate no-trade regions (which is asymptotically optimal for uncorrelated risky assets in the high risk aversion limit \cite{guasoni.muhlekarbe.15}) still outperforms a time-based rebalancing rule.

The message of these results is mixed. With constant investment opportunities, time-based rebalancing delivers virtually optimal results, with explicit formulas both for optimal trading times and their performance in arbitrary dimensions. Whence, they are an appealingly simple alternative to the less tractable optimal policies in these settings. In contrast, with stochastic investment opportunities, e.g., if optimal strategies are of trend following type, our results suggest that even ad-hoc multidimensional no-trade regions appear to be better suited than time-based rebalancing plans. Further developments in this direction -- maybe using a more refined discretization scheme\footnote{Gobet and Landon~\cite{gobet.landon.14} study the discretization error for strategies that rebalance when price increments exit an ellipsoid. Jointly minimizing transaction costs and this discretization error is more involved, requiring numerical methods to optimize functions from $\dbR^d$ to $\dbR$ derived from the expectations of hitting times.} as in Gobet and Landon~\cite{gobet.landon.14} -- therefore remain a challenging direction for future research.

The remainder of the article is organized as follows: Section~\ref{sec:model_mul} introduces the model and the optimization problem. Our main results are collected in Section~\ref{sec:results}. Subsequently, we numerically analyze the performance of our optimal strategy in different baseline models and compare it to a number of alternatives by means of Monte Carlo simulations. All proofs are collected in Appendix~\ref{heuristicmainresult}.

\paragraph{Notation} Throughout the article, $C$ stands for a generic positive constant that may vary from expression to expression. $b^\top$ denotes the transpose of a column vector $b\in \R^m$, $\|b\|_{\R^m}$ its Euclidean norm on $\R^m$, and $bc$ refers to the inner product of $b\in \R^m$ and $c\in \R^m$. 

 We write by $\mathcal{M}_{m \times d}(\R)$ for the space of $m \times d$ matrices with real entries equipped with the $L_{2,1}$- norm:
\begin{align}\label{eq:L21}
\|A\|_{2,1} := \sum_{i=1}^m \sqrt{\sum_{j=1}^d |A^{ij}|^2}.
\end{align}
For an $m \times m$ matrix $A$, $\mathrm{tr}(A) := \sum_{i=1}^m A^{ii}$ denotes its trace. Moreover, we write $\|A\|_F:=  \sqrt{\sum_{i=1}^m \sum_{j=1}^d  |A^{ij}|^2}$ for the Frobenius norm of a $m\times d$ matrix $A$ and denote by $(A^i )_{1\leq i \leq m}\in \R^d$ the i-th row of the matrix $A$.

Finally, for $E \subseteq \R^p$ and $\mathbb{K}=\R^m$ or $\mathcal{M}_{m\times d}(\R)$, we denote by $C^j_b (E; \mathbb{K})$, $j\in\mathbb{N}$ the class of $\mathbb{K}$-valued, bounded, $j$ times differentiable functions with bounded derivatives up to order $j$. 

\section{Model}\label{sec:model_mul}
\subsection{Market}
On a filtered probability space $(\Omega,\mathcal{F},\mathbb{P})$, we consider a financial market which consists of $m+1$ assets: one riskless asset, normalized to one, and $m$ risky assets driven by a $d$-dimensional Brownian motion $B=(B^1,\ldots,B^d)^\top$:
\begin{align*}
\frac{d S^i_t}{S^i_t} &= \mu^i(Y_t) dt +  \sum_{j=1}^m \sigma^{ij}(Y_t) d B^j_t, \quad i=1,\ldots,m.\label{dynamicsS1}
\end{align*}
Here, the state variable $Y=(Y^1,\ldots,Y^p)^\top \in \R^p$ is an autonomous diffusion process,
\begin{align*}
d Y^i_t &= b^i(Y_t) dt + \sum_{j=1}^p g^{ij}(Y_t) d B^j_t. 
\end{align*}
The expected excess returns of the risky assets are collected in the vector-valued process $\mu(Y)= (\mu^1(Y),\ldots,\mu^m(Y))^\top$; the entries $\sigma^{ij}(Y)$ of the matrix-valued process $\sigma(Y)$ describe the exposures of risky asset $i$ with respect to the shocks induced by the $j$-th component of the Brownian motion $B$. Likewise, $b(Y)=(b^1(Y),\ldots,b^p(Y))^\top \in \mathbb{R}^m$ and $g(Y)=(g^{ij}(Y))_{1 \leq i \leq p, 1 \leq j \leq d} \in \mathbb{R}^{p\times d}$ are the drift coefficient and diffusion matrix of the state variable. Throughout, we impose the following regularity conditions:

\begin{myassump}\label{ass:1}
Let $E \subset \R^p$ be the support of the state variable $Y$. 
\begin{enumerate}
\item $\mu \in C_b^{2}(E;\R^m)$, $b \in C_b^{2}(E;\R^p)$, $\sigma \in C^2_b(E,\mathcal{M}_{m\times d}(\R))$, and $g \in C^2_b(E,\mathcal{M}_{p\times d}(\R))$. That is, the drift and diffusion coefficients are bounded and twice continuously differentiable with bounded derivatives. In particular, $\sup_{y\in E} \|\sigma(y)\|_{2,1}<K_\sigma$ for some constant $K_\sigma>0$.

\item The covariance matrix of the risky assets, $\Sigma(Y_t) := \sigma(Y_t) \sigma(Y_t)^\top$ is uniformly elliptic, so that its inverse $\Sigma^{-1} = (\xi^{ij})_{1\leq i,j \leq m}$ exists for all $t\in[0,T]$ and is uniformly bounded on  $E$.
\end{enumerate}
\end{myassump}

\subsection{Frictionless Optimization}

To set the stage, we first briefly recapitulate the frictionless case, where portfolio rebalancing is costless. To this end, consider an investor with initial wealth $v>0$. If she holds fractions $w_t=(w_t^1,\ldots,w_t^m)^\top$ in the risky assets at time $t$, her wealth has the following dynamics:
\begin{equation*}\label{defV}
\frac{d V_t}{V_t} = \sum_{i=1}^m  w^i_t \left(\mu^i(Y_t) dt +\sigma^{i}(Y_t) d B_t\right),\quad  \quad V_0=v.
\end{equation*}
As in \cite{dufresne.al.15,guasoni.mayerhofer.14}, the investor maximizes her expected \emph{relative} returns, penalized for the corresponding variances. In the continuous-time limit, this leads to the following local mean-variance criterion:\footnote{The second equality follows from Assumption~\ref{ass:1}(ii). Similar criteria formulated in terms of the \emph{absolute} increments of the wealth process have been studied by \cite{kallsen.02} in the frictionless case and by \cite{martin.schoeneborn.11,garleanu.pedersen.13a,garleanu.pedersen.16,martin.12} with trading costs.}
\begin{equation}\label{objective}
F(w):=\frac{1}{T}\E\left[\int_0^T \frac{d V_t}{V_t} -\frac{\gamma}{2} \int_0^T\frac{d[V]_t}{V_t^2}\right]=\frac{1}{T}\int_0^T \E\left[w(Y_t)^\top \mu(Y_t) -\frac{\gamma}{2} w(Y_t)^\top \Sigma(Y_t) w(Y_t) \right]dt\to \max!
\end{equation}
Here, the risk aversion parameter $\gamma>0$ trades off the relative importances of expected returns and variances. Pointwise maximization readily yields that this objective function is maximized by the \emph{Merton portfolio}:
\begin{equation}\label{merton}
w^*(Y_t) :=(w^{*,1}(Y_t),\ldots,w^{*,m}(Y_t))^\top = \frac{1}{\gamma} \Sigma^{-1}(Y_t) \mu(Y_t).
\end{equation}
The corresponding optimal frictionless performance is obtained by plugging~\eqref{merton} back into~\eqref{objective}:
$$F(w^*)=\frac{1}{T} \E\left[\int_0^T \frac{\mu(Y_t)^\top \Sigma^{-1}(Y_t) \mu(Y_t)}{2\gamma}dt\right].$$

Throughout, we assume that the Merton portfolio short sells neither the safe nor the risky assets, and never invests all funds into one of these:

\begin{myassump}\label{ass:2}
$0 \leq w^{*,i}(Y_t) < 1$, $i=1,\ldots,m$, and  $0 < \sum_{i=1}^m w^{*,i}(Y_t) \leq 1$.
\end{myassump}

\subsection{Introducing transaction costs}

Now, we add proportional transaction costs $\varepsilon>0$ to the optimization problem~\eqref{objective}. Even with arbitrarily small transaction costs, any strategy of infinite variation leads to immediate bankruptcy. In particular, the Merton portfolio \eqref{merton} cannot be implemented, because the corresponding number of risky shares generically follows a diffusion process.\footnote{The only exceptions occur if it happens to be of buy-and-hold type because it lies in one of the corners of the unit simplex.} 

In the general multidimensional setting considered here, it is infeasible to compute the no-trade region characterizing the optimal policy, even in the limit for small transaction costs. Therefore, we focus here on the less ambitious goal of determining the optimal trading frequency within the class of ``time-based'' rebalancing strategies, which can be described as follows. At time $\tau_0=0$, set up the Merton portfolio $w^*(Y_{\tau_0})$ from \eqref{merton}, and choose the next rebalancing date $\tau_1$ based on the current market characteristics. Until time $\tau_1$, let the portfolio evolve uncontrolled, before rebalancing it back to the Merton portfolio $w^*(Y_{\tau_1})$ at $\tau_1$.\footnote{Without passing to the small-cost limit, the optimal target portfolio generally differs from the frictionless optimizer, see~\cite{garleanu.pedersen.16} for a detailed discussion in a more tractable model with quadratic costs. In the small-cost limit, the distinction disappears, and rebalancing towards the Merton portfolio becomes asymptotically optimal \cite{moreau.al.15}. This parallels asymptotic results for single-asset models with small proportional costs, where asymptotically optimal policies are also symmetric around the frictionless Merton portfolio~\cite{whalley.wilmott.97,janecek.shreve.04}. Here, we do not establish the suboptimality of other target portfolios in order to avoid the additional technicalities this would entail and since rebalancing towards the Merton portfolio is standard in practice.} Then, choose the next trading time $\tau_2$ based on the information available at time $\tau_1$, and repeat until the terminal time $T$ is reached. Since we are eventually interested in the limit $\varepsilon \downarrow 0$ of \emph{small} transaction costs, we parametrize the waiting times between successive trades as follows: 

\begin{mydef}\label{defA}
A \emph{discretization rule} is an adapted, continuous, and positive process $A$ such that 
\begin{equation}\label{intcondA}
\E\left[\int_0^T A_t dt\right]<\infty \mbox{ and } \E\left[\frac{1}{ \inf_{t\in[0,T]} A_t}\right]<\infty.
\end{equation}
The \emph{trading times} associated with the discretization rule $A$ and the transaction cost $\vep>0$ are given by the following increasing sequence of stopping times:
\begin{align*}%\label{deftradingtimes}
\tau_0=0,\quad \tau_{j}= \tau_{j-1} + \varepsilon^\alpha A_{\tau_{j-1}}, \ j=1,2,\ldots
\end{align*}
\end{mydef}

To wit, the parameter $\alpha$ governs how trading is sped up with smaller transaction costs, whereas the process $(A_t)_{t \in[0,T]}$ incorporates the current market characteristics. Notice that the second requirement in \eqref{intcondA} implies that the number of trades until maturity is a.s.~finite:
\begin{equation}\label{defN}
N:=\inf \{j\geq 0 : \tau_{j+1}\geq T\}\leq \frac{\vep^{-\alpha}T}{\inf_{t\in [0,T]} A_t}<\infty\quad \dbP-\mbox{a.s.}
\end{equation} 
Moreover, using the continuity of $A$, the first requirement in~\eqref{intcondA} yields that \begin{equation}\label{dtorder}
\lim_{\vep \rightarrow 0} \sup_{j\leq N}\{ \tau_{j+1}-\tau_{j}\}\leq \lim_{\vep \rightarrow 0}\vep^\alpha \sup_{t \in [0,T]} A_t = 0 \quad \dbP-\mbox{a.s.}
\end{equation}
That is, the mesh width of the discretization is indeed governed by $\varepsilon^\alpha$. Given a discretization rule and the initial wealth, the corresponding wealth process can be described recursively as follows:

\begin{mydef}\label{defVtac}
Fix an initial wealth $v>0$ and a discretization rule $A$. Then, the evolution of the dollar amounts $ V^{\varepsilon,i}(A)$, $i=0,\ldots,m$ invested in each of the assets and in turn the total wealth $V^\vep(A)=\sum_{i=0}^m V^{\vep,i}(A)$ evolve as follows. At the initial time $t=\tau_0=0$, we have
\begin{equation*}%\label{def:whatdynamics}
V_0^{\vep,i}(A):=v w^{*,i}(Y_0), \quad i=1,\ldots,m, \qquad V_0^{\vep,0}(A):=v-\sum_{i=1}^m V^{\vep,i}_0.
\end{equation*}
On $(\tau_{j-1},\tau_j)$ no transaction is made, so that the wealth process evolves uncontrolled as
\begin{align*}
V^{\varepsilon,0}_t(A)&=V^{\varepsilon,0}_{\tau_{j-1}}(A),\\
V_t^{\vep,i}(A)&=V^{\vep,i}_{\tau_{j-1}}(A)\exp\left(\int_{\tau_{j-1}}^t \left(\mu^i(Y_s) -\frac{1}{2}||\sigma^i(Y_s)||_{\R^d}^2\right) ds+\int_{\tau_{j-1}}^t {\sigma^i(Y_s) dB_s}\right), \quad i=1,\ldots,m,
\end{align*}
with associated risky weights $w^{\vep,i}_t(A):=V^{\vep,i}_t(A)/V^{\vep}_t(A)$, $i=0,\ldots,m$. At time $t=\tau_j$, the portfolio is rebalanced back to the Merton portfolio $w^*(Y_{\tau_j})$ from \eqref{merton}. That is, the dollar amount $V^\varepsilon_{\tau_j-} \Delta L^i_{\tau_j}$ traded in asset $i$ is determined via\footnote{That is, that the components of $\Delta L$ correspond to the fractions of wealth traded in each risky asset.}
$$ w^{\varepsilon,i}_{\tau_j}=\frac{V^\varepsilon_{\tau_j-}(A)\left(w^{\varepsilon,i}_{\tau_j-}+ \Delta L^i_{\tau_j}\right)}{V^\vep_{\tau_{j-}}(A)\left(1-\varepsilon \sum_{i=1}^m |\Delta L^i_{\tau_j}|\right)}=\frac{w^{\varepsilon,i}_{\tau_j-}+ \Delta L^i_{\tau_j}}{1-\varepsilon \sum_{i=1}^m |\Delta L^i_{\tau_j}|}\overset{!}{=}w^{*,i}_{\tau_j} ,\quad  i=1,\ldots,m.$$
Then, the risky weights $w^{\vep}_{\tau_j}(A)$ match the Merton weights $w^*(Y_{\tau_j})$ after subtracting the transaction costs from the safe account and in turn the total wealth:
\begin{align*}
V^\varepsilon_{\tau_j}(A) &= V^\varepsilon_{\tau_{j}-}(A)\left(1- \varepsilon \sum_{i=1}^m |\Delta L^i_{\tau_j}|\right).
\end{align*}
\end{mydef}

With the wealth dynamics at hand, we now formulate the investor's local mean-variance criterion with transaction costs in direct analogy to its frictionless counterpart~\eqref{objective}:\footnote{Note that both integrations in \eqref{objectivetac2} are with respect to jump processes, unlike in the frictionless case \eqref{defV}.}

\begin{equation}\label{objectivetac2}
F^\varepsilon(A):=\frac{1}{T}\E\left[\int_0^T \frac{d V_t^\varepsilon(A)}{V_{t-}^\varepsilon(A)} -\frac{\gamma}{2} \int_0^T\frac{d  [ V^\varepsilon (A)]_t}{V_{t-}^\varepsilon(A)^2}\right] \rightarrow \max!
\end{equation}

\section{Main Results}\label{sec:results}

Since the optimization problem \eqref{objectivetac2} cannot be solved in closed form, we study its limit as the transaction cost $\varepsilon$ tends to zero and the solution approaches its frictionless counterpart \eqref{merton}. Asymptotically, the objective function can then be decomposed into its frictionless counterpart, as well as losses caused directly by the transaction costs and displacement from the frictionless target, respectively (compare \cite{rogers.04,janecek.shreve.04,janecek.shreve.10,kallsen.muhlekarbe.15}): 

\begin{myprop}\label{NTACprop}
Under Assumptions~\ref{ass:1} and ~\ref{ass:2}, for\footnote{$\alpha \geq 2$ corresponds to very frequent rebalancing, for which the present asymptotics no longer apply. However, one can verify directly that the corresponding transaction costs are larger than the optimal order $O(\varepsilon^{2/3})$ in this case.} $0<\alpha<2$ and a discretization rule $A$, the objective function has the following expansion as $\varepsilon \downarrow 0$:
\begin{align}\label{decompFeps}
F^\varepsilon(A)
 =&\frac{1}{T}\E\left[ \int_0^T   \frac{\mu(Y_t)^\top \Sigma^{-1}(Y_t) \mu(Y_t)}{2 \gamma} dt \right]-\frac{1}{T}\E\left[\TAC(A) +\frac{\gamma}{2} \DE(A)\right]+ O(\vep^{2-\alpha}).
\end{align}
Here, the \emph{transaction costs} $\TAC(A)$ and the \emph{discretization error} $\DE(A)$ are given by
\begin{align}
\TAC(A) :=& \varepsilon \sum_{i=1}^m \sum_{j=1}^N |\Delta L^i_{\tau_j}|,\label{defTAC}\\
\DE(A):= &\int_0^T  (w^*(Y_t)-{w}^\varepsilon(Y_t))^\top \Sigma(Y_t) (w^*(Y_t)-{w}^\varepsilon(Y_t)) dt.\label{defDE}
\end{align}
\end{myprop}

\begin{proof} 
See Section \ref{heuprop}.
\end{proof}

In the above decomposition, the term $\TAC(A)$ tracks the transaction costs accumulated by applying the discretization rule $A$. The term $\DE(A)$ in turn measures the remaining utility loss, which is accrued due to displacement from the frictionless target portfolio. As the transaction cost $\varepsilon$ tends to zero, these terms tend to zero at different asymptotic rates determined by the parameter $\alpha$ from Definition~\ref{defA}:

\begin{mylemma}\label{lemma-return-expansion} 
Define $\beta(Y_t)= (\beta^1(Y_t),\ldots,\beta^m(Y_t))^\top$ with 
\begin{align}
\beta^i(Y_t)  &:= \tilde{\sigma}^i(Y_t)-w^{*,i}(Y_t)\left(\sigma^i(Y_t)-\sum_{k=1}^m w^{*,k}(Y_t) \sigma^k(Y_t) \right)\in \mathbb{R}^d,\label{defbeta}
\end{align}
where $\tilde \sigma (Y_t)$ is the diffusion coefficient \eqref{deftildesigma} of the Merton portfolio~\eqref{merton}.
Then, for $0<\alpha<2$ and a discretization rule $A$, the following expansions hold in the limit $\vep \downarrow 0$:
\begin{align}
\E[\TAC(A)]=&\E\left[\vep\sum_{i=1}^m \sum_{j=1}^N |\Delta L^i_{\tau_j}|\right]= \varepsilon^{1-\alpha/2} \E\left[\sqrt{\frac{2}{\pi}} \int_{0}^{T}  \frac{ \|\beta(Y_t)\|_{2,1} }{\sqrt{A_t}} dt\right]+ o(\varepsilon^{1-\alpha/2}),\label{claim-expansion1}\\
\E[\DE(A)]=&\E\left[\int_0^T  (w^*_t-{w}^\varepsilon_t)^\top \Sigma_t (w^*_t-{w}^\varepsilon_t) dt\right]=\frac{\varepsilon^\alpha}{2}\E\Biggl[\int_0^T \mathrm{tr}\left(\beta(Y_t)^\top \Sigma_t \beta(Y_t)\right)A_{t}dt \Biggr]+ o (\varepsilon^{\alpha}),\label{claim-expansion2}
\end{align}
where all the expectations are positive and finite.
\end{mylemma}

\begin{proof} 
See Section \ref{proof:32}.
\end{proof}

The crucial quantity in the above formulas is the matrix $\beta$. Its entries measure the difference between the diffusion coefficients of the frictionless target weight and the discretely rebalanced version, cf. Lemma~\ref{momentmerton} and Lemma~\ref{lemmawhat}.\footnote{Note that even though the \emph{numbers of shares} in the discretely rebalanced portfolio are of finite variation, this is not the case for the corresponding risky \emph{weights}, which fluctuate with diffusive price shocks.} 

Rebalancing more frequently evidently reduces the discretization error but also increases the incurred transaction costs. Therefore, to maximize the objective function $F^\varepsilon$ over discretization rules we have to choose  $\alpha>0$ so that the leading orders of both error terms are of the same magnitude. Lemma~\ref{lemma-return-expansion} shows that the leading orders match at $\varepsilon^{2/3}$ for $\alpha=2/3$. With this choice, maximizing the local mean variance criterion with transaction costs is -- at the leading order -- tantamount to minimizing the sum of i) transaction costs and ii) discretization error, weighted by risk aversion:\footnote{A similar criterion in terms of \emph{absolute} quantities is directly used by \cite{rosenbaum.tankov.14}.}

\begin{mydef}\label{defopt}
A discretization rule $A$ is called \emph{asymptotically optimal} if it minimizes the leading-order \emph{total cost}, 
$$\TC(A) :=\lim_{\vep \rightarrow 0 } \frac{\E\left[ \TAC(A) + \frac{\gamma}{2}\DE(A)\right]}{\vep^{2/3}}.$$
\end{mydef}

The optimal discretization rule and performance for this asymptotic criterion can be computed explicitly:

\begin{mythm}\label{mainresult_mul} 
Suppose that  Assumptions~\ref{ass:1} and \ref{ass:2} hold and $\beta$ from~\eqref{defbeta} satisfies
\begin{equation}\label{Assumptionv}
\E\left[\frac{1}{\inf_{t\in [0,T]} \|\beta(Y_t)\|^{2/3}_{2,1}}\right]<\infty.
\end{equation} 
Then, 
\begin{equation}\label{tc}
\TC(A) =\E\Biggl[\int_0^T\frac{\gamma}{4} \mathrm{tr}\left(\beta(Y_{t})^\top \Sigma(Y_t) \beta(Y_{t})\right)A_{t}dt+\sqrt{\frac{2}{\pi}} \int_{0}^{T}  \frac{ \|\beta(Y_t)\|_{2,1} }{\sqrt{A_t}} dt \Biggr].
\end{equation}
An asymptotically optimal discretization rule is given by
\begin{equation}\label{eq:optimA}
A^*(Y_t) = \left(\frac{\sqrt{\frac{2}{\pi}}\|\beta(Y_t)\|_{2,1}  }{\frac{\gamma}{2} \mathrm{tr}\left(\beta(Y_t)^\top \Sigma(Y_t)\beta(Y_t)\right)}\right)^{2/3},
\end{equation}
with associated trading times $\tau^*_0=0$ and $\tau^*_{j} = \tau_{j-1}^* + \varepsilon^{2/3} A^*(Y_{\tau_{j-1}})$, $j = 1,2, \ldots$ The corresponding minimal leading-order total cost is
\begin{equation}\label{tcformula}
\TC(A^*)=\frac{3}{2} \E\left[\int_0^T \left(\sqrt{\frac{2}{\pi}}\|\beta(Y_t)\|_{2,1}  \right)^{2/3} \left(\frac{\gamma}{2} \mathrm{tr}\left(\beta(Y_t)^\top \Sigma(Y_t)\beta(Y_t)\right)\right)^{1/3}dt\right].
\end{equation}
\end{mythm}

The asymptotically optimal trading frequency~\eqref{eq:optimA} and the corresponding welfare loss~\eqref{tcformula} relative to the frictionless case are completely determined by $\beta$ from~\eqref{defbeta} and the covariance matrix $\Sigma$ of the risky assets -- like in other models with small trading costs, the expected returns do not contribute directly at the leading order. 

The condition \eqref{Assumptionv} requires that $\beta(Y)$ is not too small. Since this term describes the difference between the diffusion coefficients of the frictionless target weight and the discretely rebalanced version, this means that the target strategy is not too close to a buy-and-hold portfolio. Otherwise, the present asymptotic regime may not apply. In multivariate Black-Scholes models this condition is satisfied as soon as one of the portfolio weights is neither zero not one. For more complex models, it needs to be verified on a case-by-case basis; we do this for a bivariate model with mean-reverting returns in Section~\ref{ss:two}.

The derivation of Formula~\eqref{tcformula} reveals that direct transaction costs contribute two thirds of the total cost $\TC(A^*)$, while the remaining one third is due to discretization error. Maybe surprisingly, these universal relative contributions do not depend on the number of risky assets, and also agree with the corresponding result for univariate asymptotically optimal move-based strategies \cite{kallsen.muhlekarbe.15}.

Mutatis mutandis, the same arguments can also be used to assess the performance of a \emph{constant rebalancing frequency}, i.e., trading times $\tau_j=\tau_{j-1}+ \varepsilon^{2/3} A$ for some constant $A>0$. These policies are even easier to interpret and implement, but lead to somewhat more cumbersome formulas. Indeed, the optimal constant discretization rule turns out to be
\begin{equation}\label{constfreq}
A^*=\left(\frac{\E\left[\sqrt{\frac{2}{\pi}}\int_0^T\|\beta(Y_t)\|_{2,1}dt\right]}{\E\left[\frac{\gamma}{2}\int_0^T \mathrm{tr}\left(\beta(Y_t)\Sigma(Y_t)\beta(Y_t)\right)dt\right]}\right)^{2/3},
\end{equation}
which leads to a total cost of 
$$\TC(A^*)= \frac{3}{2}\left(\E\left[\sqrt{\frac{2}{\pi}}\int_0^T\|\beta(Y_t)\|_{2,1}dt\right]\right)^{2/3}\left(\E\left[\frac{\gamma}{2}\int_0^T \mathrm{tr}\left(\beta(Y_t)\Sigma(Y_t)\beta(Y_t)\right)dt\right]\right)^{1/3} .$$

%--------------------------------------- Example--------------------------
\section{Examples and Implications}\label{sec:examples}

In this section, we illustrate our results with a number of examples. First, we consider the case of a single risky asset, and compare the performance of our time-based rebalancing rule with the asymptotically optimal move-based policies~\cite{guasoni.mayerhofer.14,melnyk.seifried.14}. Afterwards, we turn to models with two risky assets, where we use Monte Carlo simulations to benchmark our policies against i) a simple buy-and-hold strategy, ii) univariate move-based strategies pasted together component-wise, and iii) the optimal bivariate move-based strategy, computed numerically~\cite{max.15}. (In each simulation, the number of paths is chosen large enough to ensure that the standard error of the result is smaller than 1 basis point.) Finally, we also report the results of our asymptotic formulas for a Black-Scholes model with 10 risky assets.

%%%%%%%%%%%%%%%%%%%%%%%%%%% BLACK AND SCHOLES 1 DIM%%%%%%%%%%%%%%%%%%%%%%%%%%%%%%%%%%%%%
\subsection{Single Asset}
\paragraph{Black-Scholes Model}
We first consider the univariate Black-Scholes model, where the expected excess return $\mu$, the volatility $\sigma$ of the risky asset, and in turn the frictionless Merton portfolio $w^*=\mu/\gamma\sigma^2$ are positive constants. Then, the volatility $\tilde{\sigma}$ of $w^*$ vanishes, and we tacitly assume that $\gamma> \mu/\sigma^2$ so that Assumption~\ref{ass:2} is satisfied and \eqref{Assumptionv} holds, too, because 
$$\beta=\sigma(1-w^*)w^* \in (0,1).$$
Whence, Theorem~\ref{mainresult_mul} is applicable and yields the following asymptotically optimal discretization rule:
$$A^*= \left(\frac{\sqrt{\frac{2}{\pi}} \sigma|(1-w^*)w^*|  }{\frac{\gamma}{2}  \sigma^4(1-w^*)^2 (w^*)^2}\right)^{2/3}=\left(\frac{\sqrt{8/\pi} }{\gamma\sigma^3 |w^*(1-w^*)|}\right)^{2/3}.$$
Thus, the waiting periods between the trading times $\tau^*_{j+1}= \tau_{j}^* + \varepsilon^{2/3}A^*$ are long if i) transaction costs $\varepsilon$ are substantial, ii) the target portfolio is close to a buy-and-hold strategy, iii) risk aversion $\gamma$ is low, or iv) the market volatility $\sigma$ is low. The corresponding leading-order performance loss is given by 
$$\sigma^2 T \left( \frac{27}{8\pi} \gamma \vep^2 |w^*(1-w^*)|^4\right)^{1/3}.$$
It is equivalent to an annuity accrued in business time $\sigma^2 T= d\langle S \rangle_T/S_T^2$, and determined by risk aversion $\gamma$, transaction costs $\varepsilon$, and the term $|w^*(1-w^*)|$ which measures the frictionless optimizer's distance from a buy-and-hold strategy.\footnote{This quantity can also be interpreted as the sensitivity of the frictionless target weight with respect to relative changes of the stock price, compare~\cite{janecek.shreve.04}.} The same ingredients also appear in the asymptotically optimal move-based performance studied in \cite{guasoni.mayerhofer.14}:
$$\sigma^2 T \left(\frac{9}{32} \gamma\vep^2 |w^*(1-w^*)|^4\right)^{1/3} .$$
Whence, the asymptotic welfare loss that can be achieved by a move-based strategy differs from the optimal time-based performance by a universal factor of $\left(\frac{12}{\pi}\right)^{1/3}\approx 1.56$, independent of market and preference parameters. This complements a result of \cite{melnyk.seifried.14}, who find that, for unit relative risk aversion,\footnote{Heuristic arguments as in \cite{kallsen.muhlekarbe.15,melnyk.seifried.14} suggest that the same universal constant also applies for more general preferences.} a similar relationship holds for move-based strategies that rebalance by means of reflection off the trading boundaries, and move-based strategies that rebalance back directly to the frictionless target. Then, the universal constant linking the two respective optimizers is $2^{1/3} \approx 1.26$.

\begin{table}
\begin{center}
\begin{tabular}{cccccc}
\hline
\text{frictionless} &\text{move based}& \text{time based} & \text{buy \& hold}\\
\hline
2.50\% & 2.47\% & 2.46 \% & 2.32\% \\
\hline
\end{tabular}
\end{center}
\caption{Simulated expected profits~\eqref{objectivetac2} for different strategies.  Parameters are: sample size $N=10^6$, $dt=1/250$, $\mu=8 \%$, $\sigma=16\%$, $\gamma=5$, $T=20$, and $\varepsilon=1\%$.} 
\label{tab:performancebs1d}
\end{table}

To get a feel for the magnitude of these effects, let us consider a concrete example. For $\mu=8\%$, $\sigma=16\%$, and $\gamma=5$, the frictionless target portfolio is given by $w^*=62.5\%$. For a $1\%$ transaction cost, our time-based rebalancing rule in turn prescribes to trade once every $2.23$ years. Even for transaction costs of only $\varepsilon=0.1\%$, the waiting times are still almost 6 months. Accordingly, the welfare losses of all rebalancing strategies are relatively small. Indeed, the difference between the optimal time- and move-based rebalancing rules turns out to be negligibly small here, cf.~Table~\ref{tab:performancebs1d}.

%This point can be illustrated further by means of Lemma~\ref{lemma-return-expansion}, which allows to compare the performance of various suboptimal discretization rules. In Figure~\ref{TCvsA1d}, we plot the total cost as a function of the waiting times and find that this curve is extremely flat around its minimum. If one uses for too frequent rebalancing times, losses are compounded significantly, but substantially longer waiting periods do not worsen the performance much. This reaffirms Constantinides' point that -- with constant investment opportunities -- even very infrequent rebalancing does not lead to substantial drops in performance.
%
%
%\begin{figure}[htbp]
%\begin{center}
%\includegraphics[width=0.49\textwidth]{TCvsA1d}
%\caption{Leading order loss~\eqref{tcformula} in the one-dimensional Black-Scholes model plotted against different waiting times $\tau_{j+1}-\tau_j$. Parameters are $\mu=8\%$, $\sigma=16\%$, $\gamma=5$ and $\vep=1\%$ (Solid) respectively $\vep=0.1 \% $ (Dotted). 
%}
%\label{fig:TCvsA1d}
%\end{center}
%\end{figure}

\paragraph{Kim-Omberg Model}
Next, we pass to a setting where stochastic investment opportunities provide additional motives to trade, which makes transaction costs more important, cf.~\cite{lynch.tan.11}. To this end, we consider a variant of the model of Kim and Omberg~\cite{kim.omberg.96}, where expected returns are mean reverting. In order to satisfy our regularity assumptions, we consider a version where the expected returns are truncated for excessively large values of the state variable:
\begin{align*}
dS_t^1 =& S_t^1\left(\mu^1(Y_t) dt + \sigma dB^1_t\right),\\
dY_t=& \lambda_{Y} \left(\bar{Y}-\mu^1(Y_t)\right)dt + \alpha_{Y} \eta dB_t^1+ \alpha_Y \sqrt{1-\eta^2} dB^2_t.
\end{align*}
Here, $B= (B^1, B^2)^\top$ is two-dimensional Brownian motion and the function $\mu^1$ is a smooth cut-off version of the identity function, chosen so that $\mu^1$ is $C^2$, has bounded derivative, and $y \mapsto \mu^1(y-\bar{Y})$ is odd. That is, there exist $y_{\min}<y_{\max}$ and a small constant $\xi>0$ such that
\begin{align*}
(\mu^1)'(y) =&
\begin{cases}
1,\,\text{for } y \in [y_{\min}+\xi,y_{\max}- \xi],\\
s(y), \, \text{for } y \in (y_{\min},y_{\min}+\xi) \cup (y_{\max}-\xi,y_{\max}),\\
0,\, \text{otherwise},
\end{cases}
\end{align*} 
where $y \mapsto s(y)$ is a $(0,1)$-valued function. The state variable $Y$ follows a modified Ornstein-Uhlenbeck process with long-run mean $0<\bar Y\in (y_{\min},y_{\max})$,\footnote{To see that this is indeed the long-run mean, compute the invariant measure of the one-dimensional diffusion $Y$, which has exponential tails and is symmetric around $\bar{Y}$ because $y \mapsto \mu^1(y-\bar{Y})/\sigma^2$ is odd.} mean reversion speed $\lambda_Y >0$.  The diffusion vector of $Y$ is then given by $(\alpha_Y \eta, \alpha_Y \sqrt{1-\eta^2})$, where $\alpha_Y>0$ and $\eta\in (-1,1)$. 
 
 The additional state variable implies that the frictionless Merton solution is no longer constant but stochastic $w^*(Y_t)= \mu^1(Y_t)/\gamma \sigma^2$. We assume that the risk aversion $\gamma$ and the cut-off levels $y_{\min},y_{\max}$ for $\mu^1$ are chosen so that $\delta_1\leq w^{*} \leq \delta_2$ for $0<\delta_1<\delta_2<1$. Then, the Merton proportion $w^*(Y)$ is again an It\^o process (cf. Lemma~\ref{momentmerton}):
 \begin{align*}
 d w^{*}(Y_t)={\tilde\mu}(Y_t) dt +{\tilde\sigma}(Y_t) d B_t,
 \end{align*} 
 with
 \begin{align*}
\tilde \sigma(Y_t)= 
\begin{cases}
\frac{1}{\gamma\sigma^2} (\alpha_Y \eta, \alpha_Y \sqrt{1-\eta^2})^\top,\, &\text{for } Y_t \in [y_{\min}+\xi,y_{\max}- \xi],\\
\frac{s(Y_t)}{\gamma\sigma^2} (\alpha_Y \eta, \alpha_Y \sqrt{1-\eta^2})^\top,\, &\text{for } Y_t \in (y_{\min},y_{\min}+\xi) \cup (y_{\max}-\xi,y_{\max}),\\
0, &\,\text{otherwise}.
\end{cases}
 \end{align*}
A straightforward computation shows that the matrix $\beta(Y)$ from \eqref{defbeta} is given by
 \begin{align*}
 \beta(Y_t) =&
 \begin{cases}
 \frac{1}{\gamma\sigma^2} (\alpha_Y \eta, \alpha_Y \sqrt{1-\eta^2})^\top -w^{*}(Y_t)(1-w^{*}(Y_t)) (\sigma,0)^\top,\\
  \qquad \qquad   \qquad \qquad  \qquad  \qquad\text{for } Y_t \in [y_{\min}+\xi,y_{\max}- \xi],\\
 \frac{s(Y_t)}{\gamma\sigma^2} (\alpha_Y \eta, \alpha_Y \sqrt{1-\eta^2})^\top-w^{*}(Y_t)(1-w^{*}(Y_t)) (\sigma,0)^\top,\\
  \qquad \qquad   \qquad \qquad  \qquad  \qquad\text{for } Y_t \in  (y_{\min},y_{\min}+\xi) \cup (y_{\max}-\xi,y_{\max}),\\
 -w^{*}(Y_t)(1-w^{*}(Y_t)) (\sigma,0)^\top, \,\text{otherwise}.
 \end{cases}
 \end{align*}
In particular, there exists a $C>0$ such that
\begin{align*}
\| \beta(Y_t)\|_{2,1} \geq &
 \begin{cases}
 |\beta(Y_t)^{12}| = \frac{\alpha_Y \sqrt{1-\eta^2}}{\gamma\sigma^2}>0, \,\text{for } Y_t \in [y_{\min}+\xi,y_{\max}- \xi],\\
 \geq C \|\beta(Y_t)\|_1 = C \left(|\frac{s(Y_t)}{\gamma\sigma^2} \alpha_Y \eta -\sigma w^{*}(Y_t)(1-w^{*}(Y_t))|+|\frac{s(Y_t)}{\gamma\sigma^2}\alpha_Y \sqrt{1-\eta^2} |  \right) >0,\\
  \qquad \qquad   \qquad \qquad  \qquad  \qquad\text{for } Y_t \in (y_{\min},y_{\min}+\xi) \cup (y_{\max}-\xi,y_{\max}),\\
 |\beta(Y_t)^{11}|\geq \delta_1(1-\delta_2) \sigma>0, \,\text{otherwise},
 \end{cases}
 \end{align*}
which in turn implies that Assumption~\ref{ass:2} and Condition \eqref{Assumptionv} are satisfied. As a result, Theorem~\ref{mainresult_mul} shows that the optimal trading frequency now depends on the state variable and is given by
\begin{align*}%\label{km1dasymopt}
(\tau^*_{j}-\tau_{j-1}^*)(Y_{\tau_{j-1}})= \left(\sqrt{\frac{8}{\pi}}\frac{\vep}{\gamma \sigma^2 \|\beta (Y_{\tau_{j-1}}) \|_{2,1}}\right)^{2/3}, \quad \forall j \geq 1,
\end{align*}
where $\|\beta(y)\|_{2,1} = 2\sigma^2 (\frac{1}{2}(w^*(y)(1-w^*(y)))^2-\frac{\alpha_Y}{\gamma \sigma^3}\eta w^*(y)(1-w^*(y))+\frac{\alpha_Y^2}{2\gamma^2 \sigma^6})$ on $[y_{\min}+\xi,y_{\max}- \xi]$.

The corresponding leading-order performance loss~\eqref{tcformula} reads as follows:
\begin{align*}
\sigma^2\E\left[\int_0^T \left(\frac{27}{2 \pi} \gamma \vep^2 \|\beta(Y_t)\|_{2,1}^2\right)^{1/3} dt\right].
\end{align*}
For unit relative risk aversion, the leading-order loss of the optimal move-based strategy is again smaller by a universal factor of $(12/\pi)^{1/3}$, just like in the Black-Scholes model, \cite[Theorem 4.1]{melnyk.seifried.14}. Heuristic arguments as in \cite{kallsen.muhlekarbe.15,melnyk.seifried.14} again suggest that this relationship remains true for more general preferences.

\begin{table}
\begin{center}
\begin{tabular}{cccccc}
\hline
\text{frictionless} & \text{move based} & \text{time based} & \text{constant frequency} &\text{buy \& hold}\\
\hline
2.65\% & 2.28\% & 2.09\% & 2.09\%& 1.27\% \\
\end{tabular}
\end{center}
\caption{Simulated expected profits~\eqref{objectivetac2} for different strategies.   Parameters are taken from~\cite{barberis.00}: $\bar{Y}= 5.60\%$, $\alpha_Y= 3.68\%$, $\lambda_Y=0.2712$, $\sigma=14.28\%$, $\gamma=5$, $\eta=-0.9351$, $T=20$, $\varepsilon=1\%$, $dt=1/250$ and $N=10^6$.}
\label{tab:performancebs1dkm}
\end{table}

To illustrate these results, we consider parameters estimated from a long time series of US equity market data~\cite{barberis.00}:
$$\bar{Y}= 5.60\%, \quad \alpha_Y= 3.68\%, \quad \lambda_Y=0.2712, \quad \sigma=14.28\%, \quad \eta=-0.9351.$$
Table~\ref{tab:performancebs1dkm} collects Monte-Carlo estimates for the performances of the optimal time-based rebalancing rule, its move-based counterpart, and a simple buy-and-hold strategy. In addition, we also consider the strategy associated to the optimal \emph{constant} trading frequency~\eqref{constfreq}, which is $6.7$ months for a $1\%$ transaction cost and an investor with risk aversion $\gamma=5$. We observe that the differences between the various strategies are much more pronounced than for the Black-Scholes model, in line with the results of~\cite{lynch.tan.11}. However, the relatives magnitudes of these differences are virtually the same, and in excellent agreement with our asymptotic results. Finally, note that adapting the time-based rule to changing market characteristics only has a very small effect here; the performance of the simple constant discretization rule is virtually the same.

%%%%%%%%%%%%%%%%%%%%%%%%%%% BLACK AND SCHOLES 2 DIM%%%%%%%%%%%%%%%%%%%%%%%%%%%%%%%%%%%%%
\subsection{Two Risky Assets} \label{ss:two}
\paragraph{Black-Scholes Model} Now, we turn to a Black-Scholes model with two risky assets, with expected excess returns $\mu=(\mu^1,\mu^2) \in \mathbb{R}_{>0}^2$ and diffusion matrix 
\begin{align*}
{\sigma} &=
\left(
 \begin{matrix}
  \alpha & 0\\
  v \cdot \rho & v \sqrt{1-\rho^2}
 \end{matrix}
 \right),
\end{align*}
i.e., the risky assets have volatilities $\alpha, v >0$, respectively, and correlation $\rho \in [-1,1]$. As before, we assume that the risk aversion is sufficiently large so that the Merton portfolio 
\begin{align*}
w^* := (w^{*,1},w^{*,2})^\top&= \frac{1}{\gamma} \Sigma^{-1} \mu, \quad \text{with} \quad \Sigma = \sigma \sigma^\top,
\end{align*}
satisfies Assumption~\ref{ass:2}. The diffusion coefficient of this constant portfolio is zero; using~\eqref{defbeta} a straightforward calculation shows
\begin{align*}
\beta &=
\left(
 \begin{matrix}
  (w^{*,1}-1) w^{*,1}\alpha+ w^{*,1}w^{*,2} v \rho & w^{*,1}w^{*,2} v \sqrt{1-\rho^2}\\
  (w^{*,1}w^{*,2}\alpha+w^{*,2}(w^{*,2}-1) v \rho & w^{*,2}(w^{*,2}-1)v \sqrt{1-\rho^2}
 \end{matrix}
 \right).
\end{align*}
Therefore, we have $\|\beta\|_{2,1} \geq |\beta^{22}|>0$ if $w^{*,2}>0$  and $\|\beta\|_{2,1} = |(w^{*,1}-1) w^{*,1}\alpha|>0$ if $w^{*,2}=0$, which shows that Condition~\eqref{Assumptionv} is satisfied. Theorem~\ref{mainresult_mul} in turn yields that the optimal trading times are given by 
\begin{equation*}%\label{optBS2d}
\tau_{j+1}^*-\tau_{j}^* = \left(\frac{\varepsilon\sqrt{\frac{2}{\pi}}\|\beta\|_{2,1}  }{\frac{\gamma}{2} \cdot \mathrm{tr}\left(\beta^\top \Sigma \beta\right)}\right)^{2/3}.
\end{equation*}

\begin{figure}[tbp]
\begin{center}
\subfigure{\includegraphics[width=0.49\textwidth]{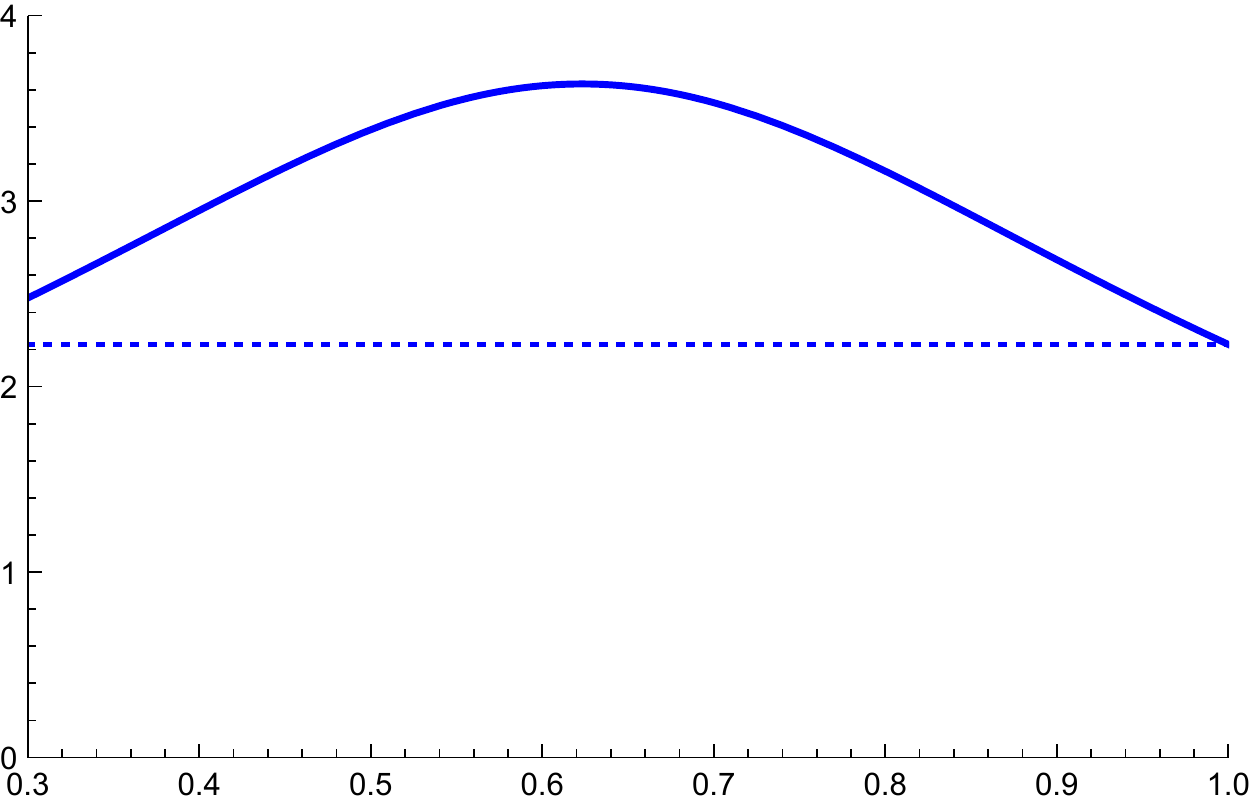}}\hfill
\subfigure{\includegraphics[width=0.49\textwidth]{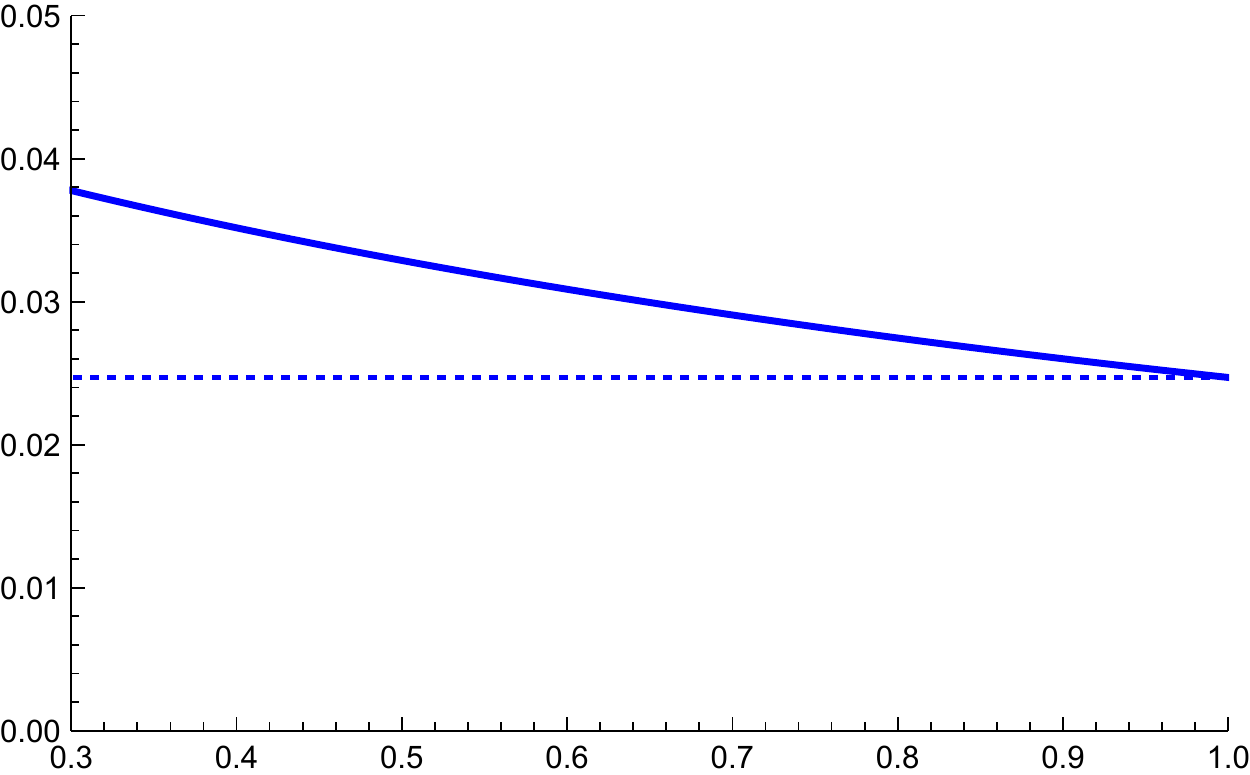}}
\caption{Left Panel: The optimal waiting times $\tau_{j}^*-\tau_{j-1}^*$ plotted for various correlations $\rho$ in the two-dimensional (solid) resp. one-dimensional (dotted) Black-Scholes model. Right Panel: Expected profits~\eqref{objectivetac2} plotted for various correlations $\rho$ in the two-dimensional (solid) resp. one-dimensional (dotted) Black-Scholes model. Parameters are $\mu^1=\mu^2=8\%$, $\alpha=v = 16\%$, $\gamma=5$ and $\varepsilon = 1\%$. 
}
\label{fig:Avsrho}
\end{center}
\end{figure}

For risky assets with identical expected excess returns and volatilities, the effect of a nonzero correlation is illustrated in Figure~\ref{fig:Avsrho} (left panel). For very large correlation ($\rho \approx 1$), the market is essentially equivalent to one with only a single risky asset. Accordingly, the optimal trading frequency converges to its univariate counterpart as $\rho \uparrow 1$. For intermediate correlations $\rho \in (0,1)$, the relationship is surprisingly non-monotonic even though the associated optimal welfare~\eqref{objectivetac2} is decreasing in $\rho$ (Figure~\ref{fig:Avsrho}, right panel).

\paragraph{Numerical Results} 

We again compare the performance of our time-based policy to a number of alternatives. The first benchmark is the asymptotically optimal move-based strategy. For more than one risky asset, explicit formulas are no longer available because one still needs to solve a free boundary problem even after passing to the small cost limit \cite{possamai.al.15}. For our illustration, we use the policy iteration algorithm proposed in \cite{max.15} to carry out these computations. The second competitor is again a simple buy-and-hold portfolio. The simulated performances are collected in Table~\ref{tab:performancebs2d}. As in the one-dimensional case, the difference between the optimal time- and move-based strategies is very small. Whence, with constant investment opportunities, our simple, explicit trading rule is  an appealing alternative to the much less tractable optimizer.

\begin{table}
\begin{center}
\begin{tabular}{cccccc}
\hline
$\rho$ & \text{frictionless} &\text{move based}  & \text{time based} & \text{buy \& hold}\\
\hline
0.3& 3.84\% & 3.80\%& 3.77\% & 3.67\%\\
0.6& 3.12\% & 3.09\%& 3.08\% & 3.01\%\\
0.9& 2.62\% & 2.60\%& 2.59\% & 2.48\%\\
\hline
\end{tabular}
\end{center}
\caption{Simulated expected profits~\eqref{objectivetac2} for different strategies.  Parameters are: sample size $N=10^6$, $dt=1/250$, $\mu^1=\mu^2=8 \%$, $\alpha=v=16\%$, $\gamma=5$, $T=20$, and $\varepsilon=1\%$.} 
\label{tab:performancebs2d}
\end{table}

\paragraph{Kim-Omberg Model}
Finally, we consider a Kim and Omberg-type model with two risky assets:
\begin{align*}
dS_t^1 =& S_t^1\left(\mu^1(Y_t) dt + \alpha dB^1_t\right),\\
dS_t^2 =& S_t^2\left(\mu^2(Y_t) dt + v \rho dB^1_t+v \sqrt{1-\rho^2} dB^2_t\right),\\
dY_t=& \lambda_{Y} \left(\bar{Y}-\mu^1(Y_t)\right)dt + \alpha_{Y} \eta dB_t^1+ \alpha_Y \sqrt{1-\eta^2} dB^2_t,
\end{align*}
Similarly as in the one-dimensional case, $B= (B^1, B^2)^\top$ is two-dimensional Brownian motion and for each $i\in\{1,2\}$ the function $\mu^i$ is a smooth truncation of the identity function with cut-off levels $y_{i,\min}$ and $y_{i,\max}$ and the associated small positive constant $\xi_i$, chosen so that $\mu^i$ is $C^2$ and has bounded derivative. Again, the state variable $Y$ follows a modified Ornstein-Uhlenbeck process with long-run mean $0<\bar Y \in (y_{1,\min},y_{1,\max})$, mean reversion speed $\lambda_Y>0$ and volatility vector $(\alpha_Y \eta, \alpha_Y \sqrt{1-\eta^2})^\top$ with $\eta\in (-1,1)$.\footnote{Note that we use the same state variable for both risky assets here, but allow for different degrees of correlations with return shocks.} The volatility matrix of the risky assets $S=(S^1,S^2)^\top$ is given by
\begin{align*}
{\sigma} &=
\left(
 \begin{matrix}
  \alpha & 0\\
  v \cdot \rho & v \sqrt{1-\rho^2}
 \end{matrix}
 \right),
\end{align*}
with $\rho \in (-1,1)$ and $0<\alpha<1/\rho$ in case of $\rho>0$, and the frictionless Merton portfolio is 
$$w^*(Y_t):= (w^{*,1}(Y_t),w^{*,2}(Y_t))^\top = \frac{1}{\gamma} \Sigma^{-1} (\mu^1(Y_t),\mu^2(Y_t))^\top,$$
with $\Sigma= \sigma \sigma^\top$ and 
\begin{align*}
\Sigma^{-1}= (\xi^{ij})_{1\leq i,j\leq 2} =&
\frac{1}{\alpha^2 v^2(1-\rho^2)}\left(
 \begin{matrix}
	 v^2 & -\rho \alpha v\\
  -\rho \alpha v & \alpha^2
 \end{matrix}
 \right).
\end{align*}
We assume that the risk aversion $\gamma$ and the cut-off levels for $\mu^1,\mu^2$ are chosen so that $$\delta_1 \leq w^{*,1}<1, \quad \delta_2\leq w^{*,2}<1, \quad 0< w^{*,1}+w^{*,2}\leq 1,$$
for some $\delta_1, \delta_2>0$. Similar calculations as before yield that for $i\in \{1,2\}$ and $y \in [y_{1,\min}+\xi_1,y_{1,\max}- \xi_1]$ the diffusion coefficients of the Merton proportion $\tilde \sigma^i$ and the matrix $\beta(Y_t)$ are given by
\begin{align*}
\tilde \sigma ^i(Y_t)= &\frac{1}{\gamma}(\xi^{i1}+\xi^{i2})(\alpha_Y \eta, \alpha_Y \sqrt{1-\eta^2})^\top\\
\beta^1(Y_t) =& \tilde{\sigma}^1(Y_t)-w^{*,1}(Y_t)((\alpha,0)^T-w^{*,1}(Y_t) (\alpha,0)^T-w^{*,2}(Y_t) (v \rho, v \sqrt{1-\rho^2})^T)\\%\label{betabsko1},\\
\beta^2(Y_t) =& \tilde{\sigma}^2(Y_t)-w^{*,2}(Y_t)((v \rho, v \sqrt{1-\rho^2})^T-w^{*,1}(Y_t) (\alpha,0)^T-w^{*,2}(Y_t) (v \rho, v \sqrt{1-\rho^2})^T).%\label{betabsko2}
\end{align*}
On $\R \backslash [y_{1,\min}+\xi_1,y_{1,\max}- \xi_1]$ similar formulas can be derived. In particular, we have for all $y \in \R$
$$\|\beta(y)\|_{2,1} \geq |\beta^{12}(y)| \geq w^{*,1}(y) w^{*,2}(y) v \sqrt{1-\rho^2}> \delta_1 \delta_2 v \sqrt{1-\rho^2} >0,$$
which in turn shows that Assumption~\ref{ass:2} and Condition~\eqref{Assumptionv} are satisfied.

\paragraph{Numerical Results}

In Table~\ref{tab:performancebs2dkm}, we again compare the performance of our time-based rebalancing rules to two alternative strategies. With mean-reverting returns, time-based rebalancing offers marked gains compared to buy-and-hold. However, the performance loss compared to the frictionless benchmark is also not negligible, similarly as in the one-dimensional case. The asymptotically optimal move-based strategy in this setting is not known explicitly, and even computing it numerically would be quite involved. As a partial remedy, we therefore consider a ``pasted'' move-based strategy, where the univariate no-trade regions are simply concatenated. This is motivated by results that link portfolio choice with high constant relative risk aversion to constant absolute risk aversion \cite{nutz.12,guasoni.muhlekarbe.15}, for which the multivariate problem factorizes for uncorrelated assets \cite{liu.04,guasoni.muhlekarbe.15}. Whence, these policies are expected to be useful proxies if relative risk aversion is sufficiently high and correlation between assets is sufficiently low. We find that for relative risk aversion $\gamma=5$, the pasted strategy outperforms the time-based rule even for high correlations, though the difference declines as correlation is increased. 

\begin{table}
\begin{center}
\begin{tabular}{cccccc}
\hline
$\rho$ & \text{frictionless} &\text{pasted} & \text{time based}& \text{buy \& hold}\\
\hline
0.3& 4.07\% & 3.53\% & 3.26\% & 2.25\%\\
0.6& 3.31\% & 2.80\% & 2.64\% & 1.79\%\\
0.9& 2.79\% & 2.27\% & 2.21\% & 1.40\%\\
\end{tabular}
\end{center}
\caption{Simulated expected profits~\eqref{objectivetac2} for different strategies.  Parameters are taken from~\cite{barberis.00}: $\bar{Y}^1=\bar{Y}^2= 5.60\%$, $\alpha_Y= v_Y= 3.68\%$, $\lambda_Y^1=\lambda_Y^2=0.2712$, $\alpha=v=14.28\%$, $\gamma=5$, $\eta=-0.9351$, $T=20$, $\varepsilon=1\%$, $dt=1/250$ and $N=10^6$.}
\label{tab:performancebs2dkm}
\end{table}

\subsection{Many Risky Assets}

Our formulas readily allow to determine optimal time-based rebalancing rules in high-dimensional contexts and analyze their performance. To illustrate this, let us consider a model with 10 risky assets. In order to sidestep issues inherent in the estimation of high-dimensional correlation matrices, we use the (positive definite) estimate provided in \cite[Table 6.5]{embrechts.al.15} for 10 assets from the DJIA. As in Section~\ref{ss:two}, we complement this with identical values for the means and variances of each asset for simplicity. Choosing $\mu=4\%$, $\sigma=20\%$ and a risk aversion of $\gamma=5$, this leads to a frictionless Merton portfolio where all risky weights are positive and a total fraction of $66\%$ is invested into the ten risky assets. In particular, Assumptions~\ref{ass:2} and \eqref{Assumptionv} are satisfied.\footnote{Higher expected returns as in Section~\ref{ss:two} would require exogenous portfolio constraints to avoid shorting the safe asset. The analysis of such constrained models is outside the scope of the present paper, but an interesting direction for further research.}

After the calculation of the corresponding matrix $\beta$, Theorem~3.4 readily yields the optimal trading frequency, which turns out to be even lower than for the one-dimensional portfolios. Here, for $\varepsilon =1\%$, it prescribes to rebalance once every 3.48 years; for a single risky asset with the same characteristics, the corresponding waiting time would be $1.84$ years. This lower trading frequency with multiple risky assets is in line with observations made in \cite{law.al.07,bichuch.guasoni.16} that optimal no-trade regions should be wider in multivariate contexts.

%------------------------------- proof---------------------------------

\begin{appendix}

\section{Proofs}\label{heuristicmainresult}
Throughout this appendix we often suppress the dependence on the state variable $Y$ to ease notation. For example, $\mu^i(Y_t)$ is frequently abbreviated to $\mu^i_t$. 

\subsection{Portfolio Dynamics}
%%%%%%%%%%%%%%%%%%%%%%%%%%%%%%%%% MERTON%%%%%%%%%%%%%%%%%%%%%%%%

First, we compute the dynamics of the Merton portfolio~\eqref{merton} and its discretized version $w^\varepsilon$ (cf. Definition~\ref{defVtac}):

\begin{mylemma}\label{momentmerton}
Under Assumptions~\ref{ass:1} and \ref{ass:2}, the Merton portfolio $w^*(Y_t)$ is an It\^o process with the following dynamics:\footnote{For better readability we omit the explicit but complicated formula of the drift $\tilde \mu (Y_t)=(\tilde \mu^1(Y_t),\ldots, \tilde \mu^m(Y_t))^T$.}
\begin{align}\label{dynamics-w}
dw^{*,i}(Y_t)={\tilde\mu}^i(Y_t) dt +{\tilde\sigma}^{i}(Y_t) d B_t, \quad i=1,\ldots,m,
\end{align}
with 
\begin{eqnarray}\label{deftildesigma}
\R^d \ni \tilde \sigma ^{i} (Y_t)= \left(\sum_{l=1}^m \sum_{k=1}^m\left(w^{*,l}(Y_t) \xi^{ik}(Y_t) (\nabla \Sigma^{kl}) G(Y_t)\right) + \frac{1}{\gamma} \xi^{il}(Y_t) (\nabla \mu^l) G(Y_t)\right)^\top,
\end{eqnarray}
where 
\begin{align*}
\nabla \Sigma^{kl} = \left(\frac{\partial \Sigma^{kl}}{\partial y_1}, \ldots, \frac{\partial \Sigma^{kl}}{\partial y_p} \right)(Y_t), \quad 
\nabla \mu^l = \left(\frac{\partial \mu^l}{\partial y_1}, \ldots, \frac{\partial \mu^l}{\partial y_p} \right)(Y_t).
\end{align*}
 Moreover, its drift and diffusion coefficients are bounded, 
\begin{align*}
\sup_{y\in E} \|\tilde{\mu}(y)\|_{\R^m}< \infty, \quad \sup_{y\in E} \|\tilde{\sigma}(y)\|_{2,1}< \infty.
\end{align*}
\end{mylemma}

\begin{proof}
Under Assumption~\ref{ass:1}, the It\^o representation and the formulas for the corresponding drift and diffusion coefficients follow from It\^o's formula. Moreover, as the Merton portfolio shorts none of the assets by Assumption~\ref{ass:2} ($0 \leq w^{*,i}_t <1$ for $i= 1,\ldots,m$ and $t\in [0,T]$), all functions appearing in~\eqref{deftildesigma} are bounded. An analogous argument applies for the drift coefficients $\tilde \mu ^i$. 
\end{proof}

%%%%%%%%%%%%%%%%%%%%%%%%%%%%% beta bound%%%%%%%%%%%%%%%%%%%%%%%%%%%%%%%%
\begin{mycor}\label{betabound}
Under Assumptions~\ref{ass:1} and \ref{ass:2}, the matrix $\beta_t = (\beta^1_t,\ldots,\beta^m_t)^\top$ with $\beta^i$ as in~\eqref{defbeta} has bounded $L_{2,1}$ norm, i.e., there exists a constant $K_\beta <\infty$ such that $\sup_{y\in E}\|\beta(y)\|_{2,1}< K_\beta.$
\end{mycor}

We now turn to the dynamics of the discretized Merton portfolio:

% % % % % % % % % % % % % % % % % % % % SDE for w^\varepsilon % % % % %
\begin{mylemma}\label{lemmawhat}
Define the rebalancing times $(\tau_j)_{j\in \mathbb{N}}$ as in Definition~\ref{defA}. Under Assumptions~\ref{ass:1} and \ref{ass:2}, the corresponding risky weights $w^\varepsilon=(w^{\varepsilon,1},\ldots,w^{\varepsilon,m})^\top$ from Definition~\ref{defVtac} satisfy: 
\begin{align}
dw^{\varepsilon,i}_t
 &= w^{\varepsilon,i}_{t}\left(\mu^{w^\varepsilon,i}_t dt+ \left( \sigma^{i}_t-\sum_{k=1}^m w^{\varepsilon,k}_{t} \sigma^k_t\right) dB_t\right) \quad \mbox{on } [\tau_{j-1},,\tau_{j}), \quad j \geq 1, \label{dynamicsweightTAC}
\end{align} 
with
\begin{align*}
\mu^{w^\varepsilon,i}_t := &  \mu^i_t-w^\varepsilon_{t} \mu_t-\sigma^i_t \sum_{k=1}^m w^{\varepsilon,k}_{t} \sigma^k_t + \left(\sum_{k=1}^m w^{\varepsilon,k}_{t} \sigma^k_t\right) \left(\sum_{k=1}^m w^{\varepsilon,k}_{t} \sigma^k_t\right), \quad i=1,\ldots, m, 
\end{align*}
and
\begin{align}
w^{\varepsilon,i}_{\tau_{j-1}}&=w^{*,i}_{\tau_{j-1}}, \quad j \geq 1.\label{dynamicsweightTACini}
\end{align}
In particular, the processes $w^{\varepsilon,i}_t$, $i=1,\ldots,m$, are well defined and take values in $[0,1)$.
\end{mylemma}

\begin{proof}
The definition of the portfolio~$w^\varepsilon$ yields \eqref{dynamicsweightTACini} and, together with It\^o's formula and the dynamics of the risky assets $S^i$ and the wealth process $V^\vep$ (cf.~Definition~\ref{defVtac}), the dynamics~\eqref{dynamicsweightTAC}. Since both the drift and diffusion coefficients in~\eqref{dynamicsweightTAC} are locally Lipschitz continuous, for the given initial value $w^*_{\tau_{j-1}}$ there exists a unique local solution $(w^\varepsilon_t)_{\tau_{j-1}\leq t \leq \tau}$ up to an explosion time $\tau$. By definition, $w^{\varepsilon,i}_{\tau_{j-1}}=w^{*,i}_{\tau_{j-1}}\in [0,1)$  for $i=0,\ldots,m$ and $j \geq 1$. Since the price process $S$ is continuous, Definition~\ref{defVtac} implies $w^{\varepsilon,i}_t\in[0,1)$ for $i\in \{0,\ldots,m\}$, $t \in [\tau_{j-1},\tau_{j})$, and $j \geq 1$. In summary, the process $w^\varepsilon$ therefore remains $[0,1)$-valued on $[0,T]$.
\end{proof}

% % % % % % % % % % % % % % % % % % % % % % PROPOSITION % % % % % % % % % % % % % %
\subsection{Proof of Proposition~\ref{NTACprop}}\label{heuprop}

\begin{proof}[Proof of Proposition~\ref{NTACprop}]

\emph{Step 1: Estimation of the Trade Sizes.} Define the rebalancing times $(\tau_j)_{j\in \mathbb{N}}$ as in Definition~\ref{defA}. Then, by Definition~\ref{defVtac}, the wealth process $V^\varepsilon$ of the corresponding strategy $w^\varepsilon$ satisfies $dV^\varepsilon_t/V_{t}^\varepsilon= \sum_{i=1}^m w^{\varepsilon,i}_{t} (\mu^i_t dt + \sigma^i_t d B_t)$ on $(\tau_{j-1},\tau_j)$, and at time $t=\tau_j$ the risky weights $(w^{\varepsilon,1}_{\tau_{j-}},...,w^{\varepsilon,m}_{\tau_{j-}})^\top$ are rebalanced back to the frictionless targets $(w^{*,1}_{\tau_j},...,w^{*,m}_{\tau_j})^\top$. Whence, for each $i=1,\ldots,m$ the respective dollar amounts $V^\varepsilon_{\tau_j-} \Delta L^i_{\tau_j}$ transferred satisfy the following rebalancing condition:
\begin{align*}
 V^\varepsilon_{\tau_{j}-}(w^{\varepsilon,i}_{\tau_{j}-}+ \Delta L^i_{\tau_j}) =& w^{*,i}_{\tau_j} V^\varepsilon_{\tau_j}
 =w^{*,i}_{\tau_j} V^\varepsilon_{\tau_j-}\left(1-\varepsilon \sum_{k=1}^m |\Delta L^k _{\tau_j}|\right).
\end{align*}
Put differently, the changes $\Delta L^i_{\tau_j}$ in the risky weights are given by
\begin{align*}
\Delta L^i_{\tau_j}+\vep w^{*,i}_{\tau_j}  \sum_{k=1}^m |\Delta L^k _{\tau_j}|=w^{*,i}_{\tau_j} -w^{\varepsilon_,i}_{\tau_{j}-}.
\end{align*}
As a consequence, the implicit function theorem for Lipschitz maps (e.g. \cite[Theorem 1]{clarke.76}) yields that, for small $\varepsilon >0$:
\begin{equation}\label{def-pi}
\Delta L^i_{\tau_j} = w^{*,i}
_{\tau_j} - w^{\varepsilon,i}_{\tau_{j}-}+P^i_{\TAC}(\tau_j).
\end{equation}
Here, under Condition~\eqref{ass:2} (no shortselling), the remainder term satisfies 
\begin{align}\label{bound-pi}
|P^i_{\TAC}(\tau_j)|\leq C\vep \| w^*_{\tau_j} - w^\varepsilon_{\tau_{j}-}\|_{{\R^m}}.
\end{align} 
In addition, 
\begin{equation}\label{jumppart}
\frac{V^\vep_{\tau_j}-V^\vep_{\tau_j-}}{V^\vep_{\tau_j-}}=-\vep \sum_{k=1}^m
|\Delta L^k_{\tau_j}|.
\end{equation}
Putting everything together we obtain
\begin{align}\label{TAChigherorder}
\int_0^T  \frac{d V_t^\varepsilon}{V_{t-}^\varepsilon}  =& \int_0^T  w^\varepsilon_t \left(\mu_t dt + \sigma_t dB_t\right)- \vep\sum_{i=1}^m \sum_{j=1}^N |w^{*,i}_{\tau_j} - w^{\varepsilon,i}_{\tau_{j}-}+P^i_{\TAC}(\tau_j)|.
\end{align}

\emph{Step 2: Estimation of the Quadratic Variation of the Jumps.}
Transaction costs are only paid at the trading times $\tau_j$. Hence, the cumulative amount of transaction costs paid up to terminal time is a pure jump process and its quadratic variation is negligible at the leading order. Indeed, \eqref{jumppart} gives
\begin{align*}
\E\left[\sum_{j=1}^N \left(\frac{V^\vep_{\tau_{j}}-V^\vep_{\tau_{j}-}}{V^\vep_{\tau_{j}-}}\right)^2 \right]= &\vep^2 \E\left[\sum_{j=1}^{N}  \left(\sum_{i=1}^m  |\Delta L^i_{\tau_j}|\right)^2 \right].
\end{align*}
By~\eqref{def-pi}, \eqref{bound-pi}, and Assumption~\ref{ass:2} (no shortselling), the changes $\Delta L^i_{\tau_j}$ of the risky weights are bounded.  Therefore, $\E[N]=O(\vep^{-\alpha})$ (cf.~\eqref{intcondA} and \eqref{defN}) implies that $$\E\left[\sum_{j=1}^m \left(\frac{V^\vep_{\tau_{j}}-V^\vep_{\tau_{j}-}}{V^\vep_{\tau_{j}-}}\right)^2 \right]=O(\vep^{2-\alpha}).$$
In addition, on $(\tau_{j-1},\tau_j)$ we have $d [V^\varepsilon_t]/(V_{t}^\varepsilon)^2= (w_t^\varepsilon)^\top \Sigma_t w^\varepsilon_{t} dt$. Thus, 
\begin{align}\label{quadvariation}
\E\left[\left|\int_0^T \frac{d [V^\varepsilon_t]}{(V_{t-}^\varepsilon)^2}-\int_0^T (w^\vep_t)^\top \Sigma_t w^\vep_{t}dt\right|\right]=O(\vep^{2-\alpha}).
\end{align}

\emph{Step 3: Expansion of the local mean-variance criterion $F^\vep$.}
Inserting~\eqref{TAChigherorder} and~\eqref{quadvariation} into the definition of $F^\vep$ (cf.~\eqref{objectivetac2}) and using that $w^*_t= \Sigma^{-1}_t\mu_t/\gamma$, the objective function simplifies to 
\begin{align*}%\label{exp-f-vep}
F^\varepsilon(A) = &\frac{1}{T}\E\left[ \int_0^T  ((w^\varepsilon_t)^\top  {\mu}_t - \frac{\gamma}{2}(w^\varepsilon_t)^\top \Sigma_t w^\varepsilon_t) dt \right]\nonumber\\
 &-  \frac{\vep}{T}\E\left[\sum_{i=1}^m \sum_{j=1}^N |w^{*,i}_{\tau_j} - w^{\varepsilon,i}_{\tau_{j}-}+P^i_{\TAC}(\tau_j)|\right]+ O(\varepsilon^{2-\alpha})\notag\\
% =& \frac{1}{T}\E\left[ \int_0^T  ((w^\varepsilon_t-w^*_t+w^*_t)^\top  {\mu}_t - \frac{\gamma}{2}(w^\varepsilon_t-w^*_t+w^*_t)^\top \Sigma_t (w^\varepsilon_t-w^*_t+w^*_t)) dt \right]\nonumber\\
%  &-  \frac{\vep}{T}\E\left[\sum_{i=1}^m \sum_{j=1}^N |w^{*,i}_{\tau_j} - w^{\varepsilon,i}_{\tau_{j}-}+P^i_{\TAC}(\tau_j)|\right]+ O(\varepsilon^{2-\alpha})\notag\\
=&\frac{1}{T}\E\left[ \int_0^T  ((w^*_t)^\top  {\mu}_t - \frac{\gamma}{2}(w^*_t)^\top \Sigma_t w^*_t) dt \right] -\frac{\gamma}{2T}\E\left[\int_0^T  (w^*_t-{w}^\varepsilon_t)^\top \Sigma_t (w^*_t-{w}^\varepsilon_t) dt\right]\notag\\
&+\frac{1}{T} \E\left[\int_0^T (w^*_t-{w}^\varepsilon_t)^\top (\mu_t-\gamma \Sigma_t w^*_t) dt\right]-  \frac{\vep}{T}\E\left[\sum_{i=1}^m \sum_{j=1}^N |w^{*,i}_{\tau_j} - w^{\varepsilon,i}_{\tau_{j}-}+P^i_{\TAC}(\tau_j)|\right]\notag\\
&+ O(\varepsilon^{2-\alpha})\notag\\
 =&\frac{1}{T}\E\left[ \int_0^T  ((w_t^*)^\top {\mu_t} - \frac{\gamma}{2}(w_t^*)^\top \Sigma_t w_t^*) dt \right]-\frac{1}{T}\E\left[\TAC(A) +\frac{\gamma}{2} \DE(A)\right]+O(\varepsilon^{2-\alpha})\nonumber\\
= &\frac{1}{T}\E\left[ \int_0^T  \frac{\mu_t^\top \Sigma^{-1}_t \mu_t}{2 \gamma} dt \right]-\frac{1}{T}\E\left[\TAC(A) +\frac{\gamma}{2} \DE(A)\right]+O(\varepsilon^{2-\alpha}). 
\end{align*}
This completes the proof of Proposition~\ref{NTACprop}.
\end{proof}

\subsection{Proof of Lemma~\ref{lemma-return-expansion}}\label{proof:32}
To establish Lemma~\ref{lemma-return-expansion}, we first derive the following useful asymptotic result:

\begin{mylemma}\label{lemmahigheroderdt}
Define the rebalancing times $(\tau_j)_{j=0,1,\ldots}$ as in Definition~\ref{defA} and let $(X^\vep)_{\vep >0}$ be a family of right-continuous, $m$-dimensional processes such that:
\begin{itemize}
\item[(i)] The process $X^\varepsilon$ satisfies $dX^\vep_t=\Theta^\vep_t dt+\Gamma^\vep_t dB_t$ for $t\in(\tau_{j-1},\tau_j)$. 
\item[(ii)] $X^\vep_{\tau_j}=0$ for all $j$.
\item[(ii)] There exists a constant $C>0$ such that
\begin{align*}
\sup_{t\in[0,T]} \left(\|\Theta^\vep_t\|_{\R^m}+ \|\Gamma^\vep_t\|_{2,1}\right)\leq C, \quad \lim_{\vep \downarrow 0}\E\left[\sup_{t\in[0,T]}\|\Gamma^\vep_t\|^2_{2,1}\right] = 0.
\end{align*}
\end{itemize}
Then:
$$\E\left[\sum_{j=1}^N \|X^\vep_{\tau_j-}-X^\vep_{\tau_{j-1}}\|_{\R^m}\right]=o(\vep^{-\alpha/2}).$$
\end{mylemma}

\begin{proof}
The Burkholder-Davis-Gundy inequality yields
\begin{align*}
\E\left[\sum_{j=1}^N \|X^\vep_{\tau_j-}-X^\vep_{\tau_{j-1}}\|_{\R^m}\right]=&\E\left[\sum_{j=1}^N \left\|\int_{\tau_{j-1}}^{\tau_{j}}\Theta^\vep_t dt+\Gamma^\vep_t dB_t\right\|_{\R^m}\right]\\
\leq& \E\left[\sum_{j=1}^N \int_{\tau_{j-1}}^{\tau_j} \|\Theta^\vep_t\|_{\R^m} dt\right]+\E\left[\sum_{j=1}^N\left\|\int_{\tau_{j-1}}^{\tau_j}\Gamma^\vep_t dB_t\right\|_{\R^m}\right]\\
\leq& C\E\left[\sum_{j=1}^N (\tau_j-\tau_{j-1})\right]+C\E\left[\sup_{t\in[0,T]}\|\Gamma^\vep_t\|_{2,1} \sum_{j=1}^N ({\tau_j-\tau_{j-1}} )^{1/2}\right]\\
\leq & C+C\vep^{-\alpha/2}\E\left[\sup_{t\in[0,T]}\|\Gamma^\vep_t\|_{2,1}   \frac{T}{\inf_{t\in[0,T]} A_t^{1/2}}\right].
\end{align*}
Here, we have used in the last step that $\sum_{j=1}^N(\tau_j-\tau_{j-1})^{1/2}=\sum_{j=1}^N\frac{\tau_j-\tau_{j-1}}{\vep^{\alpha/2}A_{\tau_{j-1}}^{1/2}} \leq  \frac{T}{\vep^{\alpha/2}\inf_{t\in[0,T]} A_t^{1/2}}$, and conclude using H\"older's inequality and the assumptions on $A$ and $\Gamma^\vep$.
\end{proof}

By combining Lemma~\ref{lemmahigheroderdt} with elementary estimates for the normal distribution, we can now establish Lemma~\ref{lemma-return-expansion}:

\begin{proof}[Proof of Lemma~\ref{lemma-return-expansion}]
The finiteness of the expectations follows from the integrability conditions on $A$ and the boundedness of $\Sigma$ and $\beta$.

\emph{Step 1: Expansion of the Transaction Cost Loss.}
We now perform a more detailed analysis of the leading order term $w^{*,i}_{\tau_{j}}-w^{\varepsilon,i}_{\tau_j-}$ in~\eqref{def-pi}. Using the dynamics of $w^*$ and $w^\varepsilon$ in~\eqref{dynamics-w} and \eqref{dynamicsweightTAC}, respectively, we find
\begin{align}
w^{*,i}_{\tau_{j}}-w^{\varepsilon,i}_{\tau_j-}=& \int_{\tau_{j-1}}^{\tau_j} \left(\tilde{\mu}^i_t-w^{\varepsilon,i}_{t}\mu^{w^\varepsilon,i}_t\right)dt+\int_{\tau_{j-1}}^{\tau_j}\left[\tilde \sigma^i_t-w^{\varepsilon,i}_{t}\left( \sigma^{i}_t-\sum_{k=1}^m w^{\varepsilon,k}_{t}\sigma^k_t\right)\right] dB_t.\label{TACtransform1}
\end{align}
Adding and subtracting $\int_{\tau_{j-1}}^{\tau_j} \beta^i_{\tau_{j-1}} dB_t$ in~\eqref{TACtransform1} yields
\begin{align*}
w^{*,i}_{\tau_{j}}-w^{\varepsilon,i}_{\tau_j-}=& \beta^i_{\tau_{j-1}} (B_{\tau_j}-B_{\tau_{j-1}}) + R^i_{\TAC}(\tau_j),
\end{align*}
where $R^i_{\TAC}(\tau_j)$ is defined as follows:
\begin{align}\label{def-ri}
R^i_{\TAC}(\tau_j):=& \int_{\tau_{j-1}}^{\tau_j} J^{FV,i}_t dt+ \int_{\tau_{j-1}}^{\tau_j} J^{S,i}_t d B_t,
\end{align}
with 
\begin{align*}
J^{FV,i}_t=\tilde{\mu}^i_t-w^{\varepsilon,i}_{t} \mu^{w^\varepsilon,i}_{t}, \quad 
J^{S,i}_t =\tilde \sigma^i_t-w^{\varepsilon,i}_{t}\left( \sigma^{i}_t-\sum_{k=1}^m w^{\varepsilon,k}_{t} \sigma^k_t\right)-\beta^i_{\tau_{j-1}}, \quad \mbox{for } t\in(\tau_{j-1},\tau_j).
\end{align*}
In vector respectively matrix notation we have $J^{FV}= (J^{FV,1},\ldots,J^{FV,m})^\top \in \R^m$ and $J^S=(J^{S,1},\ldots,J^{S,m})^\top \in \mathcal{M}_{m\times d} (\R)$.
Lemma~\ref{lemmawhat} and Lemma~\ref{momentmerton} imply that $J^{FV}$ is bounded. Using that the mapping 
$$w\to \left[\tilde \sigma^i_t-w^i\left( \sigma^{i}_t-\sum_{k=1}^m w^k \sigma^k_t\right)\right]$$ 
is locally Lipschitz continuous with Lipschitz constant $L$ depending on $m$ and $K_\sigma$ (cf.~Assumption~\ref{ass:1}(ii)), we obtain that 
\begin{align*}%\label{bound-ri}
\|J^{S,i}_t \|_{\R^d}\leq L\| w^\vep_t-w^*_t\|_{\R^m} \mbox { and }|J^{FV,i}_t|\leq C \mbox{ on }[\tau_{j-1},\tau_j),
\end{align*}
where $C$ denotes a real constant. In view of~\eqref{def-pi} we also have
$$\vep\E\left[\sum_{i=1}^m \sum_{j=1}^N |\Delta L^i_{\tau_j}|\right]=\vep\E\left[\sum_{i=1}^m \sum_{j=1}^N |\beta^i_{\tau_{j-1}} (B_{\tau_j}-B_{\tau_{j-1}})+R^i_{\TAC} (\tau_j)+P^i_{\TAC}(\tau_j)|\right],$$
where $R^i$ and $P^i$ are defined as in \eqref{def-ri} and \eqref{def-pi}, respectively. 

In order to prove the expansion~\eqref{claim-expansion1} it suffices to show that:
\begin{align}
\vep\E\left[ \sum_{i=1}^m \sum_{j=1}^N |\beta^i_{\tau_{j-1}} (B_{\tau_j}-B_{\tau_{j-1}})|\right]=& \varepsilon^{1-\alpha/2} \E\left[\sqrt{\frac{2}{\pi}} \int_{0}^{T}  \frac{ \|\beta_t\|_{2,1} }{\sqrt{A_t}} dt\right]+ o(\varepsilon^{1-\alpha/2}),\label{expansion-return1}
\end{align}
and
\begin{align}
\vep\E\left[\sum_{i=1}^m \sum_{j=1}^N |R^i_{\TAC} (\tau_j)|+ |P^i_{\TAC}(\tau_j)|\right]=& o(\varepsilon^{1-\alpha/2}).\label{expansion-return2}
\end{align}
For~\eqref{expansion-return1} we use the independent increments property of Brownian motion and the scaling property of the normal distribution to obtain
\begin{align*}
\vep\E\left[ \sum_{i=1}^m \sum_{j=1}^N |\beta^i_{\tau_{j-1}} (B_{\tau_j}-B_{\tau_{j-1}})|\right] 
 &= \vep\E\left[ \sum_{j=1}^\infty \sum_{i=1}^m \mathbf{1} _{\{\tau_{j}<T\}}|\beta^i_{\tau_{j-1}}( B_{\tau_j}-B_{\tau_{j-1}})|\right]\\
 &= \vep\E\left[  \E\left[\sum_{j=1}^\infty \sum_{i=1}^m \mathbf{1} _{\{\tau_{j}<T\}}|\beta^i_{\tau_{j-1}}( B_{\tau_j}-B_{\tau_{j-1}})| \Biggl| \mathcal{F}_{\tau_{j-1}}\right]\right]\\
 &=\vep\E\left[\sum_{j=1}^\infty \mathbf{1} _{\{\tau_{j}<T\}}\sum_{i=1}^m   \|\beta^i_{\tau_{j-1}}\|_{\R^d} |Z| \sqrt{\tau_{j}-\tau_{j-1}}\right]\\
 &= \vep\E\left[\sum_{\substack{0<\tau_j < T}} \sum_{i=1}^m  \sqrt{\frac{2}{\pi}} \|\beta^i_{\tau_{j-1}}\|_{\R^d}  \sqrt{\varepsilon^\alpha A_{\tau_{j-1}}}\right].
\end{align*}
Here, $Z$ denotes an independent univariate standard normal random variable, for which $\E[|Z|] = \sqrt{2/\pi}$. To approximate the random sum over $\tau_j$ we rewrite the expression inside the sum:
\begin{align*}
\E[\TAC(A)] 
 &=\vep\E\left[\sum_{\substack{0<\tau_j< T}} \sum_{i=1}^m  \sqrt{\frac{2}{\pi}} \|\beta^i_{\tau_{j-1}}\|_{\R^d}\frac{\varepsilon^\alpha A_{\tau_{j-1}}}{\sqrt{\varepsilon^\alpha A_{\tau_{j-1}}}}\right]\nonumber\\
 &= \E\left[\sum_{\substack{0<\tau_j< T}} \sum_{i=1}^m  \sqrt{\frac{2}{\pi}}\varepsilon^{1-\alpha/2}  \|\beta^i_{\tau_{j-1}}\|_{\R^d} \frac{\tau_{j}-\tau_{j-1}}{\sqrt{ A_{\tau_{j-1}}}}\right].\nonumber
 \end{align*}
Since $\sup_{y \in E}\|\beta(y)\|_{2,1}$ is bounded by Corollary~\ref{betabound}, the lower bound on $A$ from Definition~\ref{defA} and the estimate~\eqref{dtorder} imply
 \begin{align}
\E[\TAC(A)] =& \E\Biggl[\sqrt{\frac{2}{\pi}}\varepsilon^{1-\alpha/2} \int_{0}^{T}   \frac{  \|\beta_{t}\|_{2,1}}{\sqrt{A_t}} dt\nonumber\\
 &\qquad + \sqrt{\frac{2}{\pi}}\varepsilon^{1-\alpha/2} \Bigg( \sum_{\substack{0<\tau_j < T}}\frac{\|\beta_{\tau_{j-1}}\|_{2,1}}{\sqrt{A_{\tau_{j-1}}}} (\tau_j-\tau_{j-1})-\int_{0}^{T} \frac{\|\beta_t\|_{2,1}}{\sqrt{A_t}}  dt\Bigg) \Biggr]\nonumber\\
 =& \E\left[\sqrt{\frac{2}{\pi}}\varepsilon^{1-\alpha/2} \int_{0}^{T}  \frac{ \|\beta_t\|_{2,1} }{\sqrt{A_t}} dt\right]+ o(\varepsilon^{1-\alpha/2}).\label{TACsol}
\end{align}
Here, dominated convergence is applicable in the last step due to the boundedness of $\beta$ and the integrability of $1/A$. 

We now turn to \eqref{expansion-return2} and consider the first term in the corresponding expectation. To estimate it, we apply Lemma \ref{lemmahigheroderdt} with 
 $$X^\vep_t:= \int_{\tau_{j-1}}^{t} J^{FV}_s ds+ \int_{\tau_{j-1}}^{t} J^{S}_s d B_s \quad \mbox{ on }[\tau_{j-1},\tau_j).$$
In view of the uniform boundedness of $J^S$ and $J^{FV}$ and because 
$$\E\left[\sup_{s\in[0,T]} \|J^{S}_s\|^2_{2,1}\right]\leq 2^{m-1} L^2 \E\left[\sup_{s\in[0,T]} \| w^\vep_s-w^*_s\|^2_{\R^m}\right],$$
it then remains to show $\E[\sup_{s\in[0,T]} \| w^\vep_s-w^*_s\|^2_{\R^m}] \rightarrow0$ as $\vep\downarrow 0$.
To this end, recall that
$$ w_t^\vep-w_t^*= \int_{\tau_{j-1}}^{t} J^{FV}_s ds+ \int_{\tau_{j-1}}^t  \left(J^{S}_s +\beta_{\tau_{j-1}}\right) d B_s, \quad \mbox{for } t\in[\tau_{j-1},\tau_{j}).$$

We define the following continuous semimartingale on $[0,T]$, 
\begin{align*}
&\tilde X^\vep_0=w_0^\vep-w_0^*,\mbox{ and for j=1\ldots N},\\
&\tilde X^\vep_t=\tilde X^\vep_{\tau_{j-1}}+ \int_{\tau_{j-1}}^{t} J^{FV}_s ds+\int_{\tau_{j-1}}^t  \left(J^{S}_s +\beta_{\tau_{j-1}}\right) d B_s, \quad \mbox{for } t\in[\tau_{j-1},\tau_{j}).
\end{align*}

For all $\omega\in \Omega$, denote by $\tilde K^\vep(\omega)$ the $1/16-$H\"older constant of $\tilde X^\vep(\omega)$ on $[0,T]$. Note that $\tilde X^\vep$ and $w^\vep-w^*$ have the same Holder constant on $(\tau_{j-1},\tau_j)$. Using the Burkholder-Davis-Gundy inequality and the uniform boundedness of $J^{FV}$ and $J^S$, it follows that
 \begin{align}\label{kolmogorovbound}
 \E\left[\|\tilde X^\vep_t-\tilde X^\vep_s\|_{\R^m}^4\right]\leq C(t-s)^2.
\end{align}  
Kolmogorov's H\"older continuity criterion (e.g. \cite[Corollary 1.2]{walsh.84}) in turn shows that $\E\left[(\tilde K^\vep)^4\right]<C$ for a constant $C$ that does not depend on $\vep$.
To prove $\E\left[\sup_{s\in[0,T]}\|w^\vep_{s}-w^*_{s}\|^2_{\R^m} \right]\rightarrow 0\mbox{ as }\vep\rightarrow 0,$ we fix an arbitrary small $r>0$ and find a sufficiently large $M_r>0$ such that the set $\Omega_r:=\{\sup_{s} A_s^{1/8}\leq M_r\}$ satisfies 
$$\dbP(\Omega_r)\geq 1-\frac{r}{2m}.$$
Then, since for each $i = 1,\dots,m$ the processes $w^{*,i}$ and $w^{\vep,i}$ are $[0,1)$-valued, using the definition of the trading times we obtain
\begin{align*}
\E\left[\sup_{s\in[0,T]}\|w^\vep_{s}-w^*_{s}\|^2_{\R^m} \right]=&\E\left[\sup_{s\in[0,T]}\|w^\vep_{s}-w^*_{s}\|^2_{\R^m} {\mathbf{1}}_{\Omega_r}\right]+\E\left[\sup_{s\in[0,T]}\|w^\vep_{s}-w^*_{s}\|^2_{\R^m} {\mathbf{1}}_{\Omega_r^c}\right]\\
 \leq &\E\left[(\tilde K ^\vep)^2 \ \vep^{\alpha/8}\sup_{s\in[0,T]} A_s^{1/8} \ {\mathbf{1}}_{\Omega_r}\right]+\E\left[\sup_{s\in[0,T]}\|w^\vep_{s}-w^*_{s}\|^2_{\R^m} {\mathbf{1}}_{\Omega_r^c}\right]\\
\leq& M_r \vep^{\alpha/8}\E[(\tilde{K}^{\vep})^2]+r/2.
 \end{align*}
Due to the uniform bound on $\E[(\tilde{K}^{\vep})^4]$, we can pick $\vep>0$ small enough to ensure $M_r \vep^{\alpha/8}\E[(\tilde{K}^{\vep})^2]\leq r/2$ and in turn
 $$\E\left[\sup_{s\in[0,T]}\|w^\vep_{s}-w^*_{s}\|^2_{\R^m} \right] \leq r.$$
 As $r$ was arbitrary, this shows that Lemma~\ref{lemmahigheroderdt} is applicable for $X^\vep$ and yields
\begin{equation*}%\label{Rterm}
\vep\E\left[\sum_{i=1}^m \sum_{j=1}^N |R^i_{\TAC} (\tau_j)|\right]=o(\varepsilon^{1-\alpha/2}).
\end{equation*}
Therefore, the first term in \eqref{expansion-return2} is indeed of the claimed asymptotic order. As for the second, notice that for $\xi,\eta \in \R$ the triangle inequalities $|\xi|-|\eta|\leq |\xi+\eta| \leq |\xi|+|\eta|$ and~\eqref{expansion-return1} show that
\begin{align}
\lim_{\varepsilon \rightarrow 0}\frac{\vep\E\left[\sum_{i=1}^m \sum_{j=1}^N |w^{*,i}_{\tau_j}-w^{\vep,i}_{\tau_j-}|\right]}{\vep\E\left[\sum_{i=1}^m \sum_{j=1}^N |\beta^i_{\tau_{j-1}} (B_{\tau_j}-B_{\tau_{j-1}})\right]} =& \lim_{\varepsilon \rightarrow 0}\frac{\vep \E\left[\sum_{i=1}^m \sum_{j=1}^N |\beta^i_{\tau_{j-1}} (B_{\tau_j}-B_{\tau_{j-1}})+R^i_{\TAC} (\tau_j)|\right]}{\vep \E\left[\sum_{i=1}^m \sum_{j=1}^N |\beta^i_{\tau_{j-1}} (B_{\tau_j}-B_{\tau_{j-1}})\right]}\nonumber\\
 =& 1.\label{MandR}
\end{align}
Moreover, the inequality~\eqref{bound-pi} gives
$$\vep \E\left[\sum_{i=1}^m \sum_{j=1}^N |P^i_{\TAC} (\tau_j)|\right] \leq C m \vep^2 \E\left[\sum_{i=1}^m \sum_{j=1}^N |w^{*,i}_{\tau_j}-w^{\vep,i}_{\tau_j-}|\right].$$ 
Therefore, 
\begin{align*}
\lim_{\varepsilon \rightarrow 0}\frac{\vep\E\left[\sum_{i=1}^m \sum_{j=1}^N |P^i_{\TAC} (\tau_j)|\right]}{\varepsilon^{1-\alpha/2}} \leq& C m \lim_{\varepsilon \rightarrow 0} \vep \frac{\vep \E\left[\sum_{i=1}^m \sum_{j=1}^N |w^{*,i}_{\tau_j}-w^{\vep,i}_{\tau_j-}|\right]}{\vep \E\left[\sum_{i=1}^m \sum_{j=1}^N |\beta^i_{\tau_{j-1}} (B_{\tau_j}-B_{\tau_{j-1}})\right]} \\
 & \qquad \times \frac{\vep \E\left[\sum_{i=1}^m \sum_{j=1}^N |\beta^i_{\tau_{j-1}} (B_{\tau_j}-B_{\tau_{j-1}})\right]}{\varepsilon^{1-\alpha/2}}\\
 =& 0,
\end{align*}
where we have used~\eqref{MandR} and~\eqref{expansion-return1} in the last step. This
completes the proof of \eqref{expansion-return2} and in turn \eqref{claim-expansion1}.

\emph{Step 2: Expansion of the Discretization Error.} 
For the discretization error we proceed similarly. For $j\in \mathbb{N}$ and $r\in [\tau_{j-1},\tau_j)$ recall that
$$w^{*,i}_r-w^{\varepsilon,i}_r = \int_{\tau_{j-1}}^{r} \left(\tilde{\mu}^i_t-w^{\varepsilon,i}_t \mu^{w^\varepsilon,i}_t\right)dt + \left[\tilde \sigma^i_t-w^{\varepsilon,i}_t\left( \sigma^{i}_t-\sum_{k=1}^m w^{\varepsilon,k}_t \sigma^k_t\right)\right] dB_t.$$ 
Arguing similarly as in Step 1, we can replace $w^*_t-w^\varepsilon_t$ and $\Sigma_t$ at the leading order with $\int_{\tau_{j-1}}^{r}\beta_{\tau_{j-1}} d B_u \in \R^m$ and $\Sigma_{\tau_{j-1}}$, respectively. Therefore, the discretization error can be rewritten as
\begin{align*}
\E[\DE(A)] &=\E\left[\int_0^T  (w_t^*-{w}_t^\varepsilon)^\top \Sigma_t (w_t^*-{w}_t^\varepsilon) dt\right]\\
 &=\E\left[\sum_{0<\tau_j < T} \int_{\tau_{j-1}}^{\tau_j} \left(\int_{\tau_{j-1}}^{t}\beta_{\tau_{j-1}}  d B_u\right)^\top  \Sigma_{\tau_{j-1}} \left(\int_{\tau_{j-1}}^{t}\beta_{\tau_{j-1}}  d B_u\right) dt\right]+o(\varepsilon^{\alpha}).
   \end{align*}
Conditioning on $\mathcal{F}_{\tau_{j-1}}$ and integrating over $t$ in turn yields
 \begin{align*}
\E[\DE(A)]  &=\E\left[\sum_{0<\tau_j < T} \int_{\tau_{j-1}}^{\tau_j} \mathrm{tr}\left(\beta_{\tau_{j-1}}^\top \Sigma_{\tau_{j-1}}\beta_{\tau_{j-1}}\right)(t-\tau_{j-1}) dt \right]+o(\varepsilon^{\alpha})\\
 &=\frac{1}{2}\E\left[\sum_{0<\tau_j < T} \mathrm{tr}\left(\beta_{\tau_{j-1}}^\top \Sigma_{\tau_{j-1}} \beta_{\tau_{j-1}}\right)(\tau_j-\tau_{j-1})^2 \right]+o(\varepsilon^{\alpha})\\
&=\frac{\varepsilon^\alpha}{2}\E\Biggl[\sum_{0<\tau_j  < T} \mathrm{tr}\left(\beta_{\tau_{j-1}}^\top \Sigma_{\tau_{j-1}}\beta_{\tau_{j-1}}\right)A_{\tau_{j-1}}(\tau_j-\tau_{j-1}) \Biggr]+o(\varepsilon^{\alpha}).
\end{align*}
Since the processes $\Sigma$ and $\beta$ are assumed to be bounded and continuous and $A$ satisfies~\eqref{intcondA}, applying the same dominated convergence argument as before implies that the sum can again be approximated by the corresponding integral:
\begin{align*}
\E[\DE(A)]  =\frac{\varepsilon^\alpha}{2}\E\Biggl[\int_0^T \mathrm{tr}\left(\beta_{t}^\top \Sigma_t \beta_{{t}}\right)A_{t}dt \Biggr]+ o (\varepsilon^{\alpha}).
\end{align*}
This completes the proof of \eqref{claim-expansion2} and in turn Lemma~\ref{lemma-return-expansion}.
\end{proof}

\subsection{Proof of Theorem~\ref{mainresult_mul}}
\begin{proof}[Proof of Theorem~\ref{mainresult_mul}]
%To minimize the total loss, we have to match the two leading orders in the asymptotic formulas (\ref{claim-expansion1}-\ref{claim-expansion2}) for $\TAC$ and $\DE$:
%$$1-\frac{\alpha}{2}= \alpha \Leftrightarrow \alpha = \frac{2}{3}.$$
With the choice $\alpha = \frac{2}{3}$, the asymptotic expansion of $F^\vep$ simplifies to
\begin{align*}
F^\varepsilon(A) &=\frac{1}{T}\E\left[ \int_0^T  \frac{\mu_t^\top \Sigma^{-1}_t \mu_t}{2 \gamma} dt \right]-\frac{\varepsilon^{2/3}}{T}\E\Biggl[\int_0^T\frac{\gamma}{4} \mathrm{tr}\left(\beta_{t}^\top \Sigma_t \beta_{{t}}\right)A_{t}dt+\sqrt{\frac{2}{\pi}} \int_{0}^{T}  \frac{ \|\beta_t\|_{2,1} }{\sqrt{A_t}} dt \Biggr]+o(\varepsilon^{2/3})\notag.
 \end{align*}
To obtain the optimal discretization rule $A$, it therefore remains to solve the following optimization problem:
\begin{equation}\label{costfunctional}
\min_{(A_t)_{t \in [0,T]}} \varepsilon^{2/3}  \E\left[\int_{0}^{T}  \frac{\gamma}{4} \mathrm{tr}\left(\beta_{t}^\top \Sigma_t \beta_{t}\right)A_{t} + \sqrt{\frac{2}{\pi}}  \frac{ \|\beta_t\|_{2,1} }{\sqrt{A_t}} dt\right].
\end{equation}
\emph{Pointwise} minimization of the integrand readily yields the optimizer from~\eqref{eq:optimA}:
\begin{equation*}%\label{Astarsol}
A^*_t = \left(\frac{\sqrt{\frac{2}{\pi}}\|\beta_t\|_{2,1}  }{\frac{\gamma}{2} \mathrm{tr}\left(\beta_{t}^\top \Sigma_t \beta_{t}\right)}\right)^{2/3}.
\end{equation*}
Plugging $A^*$ back into the cost functional~\eqref{costfunctional} in turn yields that the corresponding minimal leading-order total cost is given by the formula reported in~\eqref{tcformula}.

To complete the proof, it remains to show that the process $A^*$ is admissible in the sense of Definition~\ref{defA}. To this end, notice that the sub-multiplicativity of the Frobenius norm yields
\begin{equation*}
\mathrm{tr}\left(\beta^\top_t \Sigma_t \beta_t\right) = \mathrm{tr}\left(\beta_t^\top \sigma_t \sigma^\top_t \beta_t\right)= \| \beta_t \sigma_t\|_F^2 \leq \|\beta_t\|_F^2 \|\sigma_t\|_F^2 \leq \|\beta_t\|_{2,1}^2 \|\sigma_t\|_{2,1}^2.
\end{equation*}
In view of Corollary~\ref{betabound}, it follows that
$$A^*_t \geq C \left(\frac{\|\beta_t\|_{2,1}}{\|\beta_t\|_{2,1}^2 \|\sigma_t\|_{2,1}^2}\right)^{2/3} \geq C \left(\frac{1}{K_\sigma^2 K_\beta}\right)^{2/3} >0, \quad \text{for all $t \in [0,T]$.} $$
Therefore, $\E[(\inf_{t\in [0,T]} A^*_t)^{-1}]<\infty$. Moreover, the uniform ellipticity of the covariance matrix $\Sigma$ implies that 
\begin{align*}
\mathrm{tr}\left(\beta_t^\top \Sigma_t \beta\right) \geq & C\ \mathrm{tr}\left(\beta_t^\top \beta_t\right) = C \|\beta_t\|_F^2 \geq C \|\beta_t\|_{2,1}^2.
\end{align*}
As a result:
$$A^*_t \leq C \left(\frac{\|\beta_t\|_{2,1}}{\|\beta_t\|_{2,1}^2}\right)^{2/3} \leq  \frac{C}{\|\beta_t\|_{2,1}^{2/3}},$$
so that $\int_0^T A^*_t dt$ has finite expectation by Assumption~\eqref{Assumptionv}. In summary, the discretization rule $A^*$ is admissible and therefore indeed asymptotically optimal.
\end{proof}

\end{appendix}

%----------------------------------------BIB----------------------------------------------------------------------------------------------------------------------
\bibliographystyle{abbrv}
\bibliography{multiTAC}

\end{document}